\documentclass[11pt]{article} 

\usepackage{fullpage}

\usepackage{amsmath,amsthm, amssymb}
\usepackage{latexsym}
\usepackage{graphicx}
\usepackage{ifthen}
\usepackage{subcaption}
\usepackage{multicol}
\usepackage{algorithm, algorithmic}
\usepackage[numbers]{natbib}
\usepackage{mathtools}
\usepackage{dsfont}
\usepackage{nicefrac}
\usepackage{wrapfig}
\usepackage{xcolor}

\usepackage{array}   
\newcolumntype{L}{>{$}l<{$}} 

\usepackage{hyperref}
\hypersetup{colorlinks=true,linkcolor=blue,citecolor=blue,urlcolor=blue}

\newcommand{\INPUT}{\item[{\bf Input:}]}

\newtheorem*{theorem*}{Theorem}
\newtheorem*{definition*}{Definition}

\newtheorem{lemma}{Lemma}

\newtheorem{definition}{Definition}

\newtheorem{claim}{Claim}

\newtheorem*{thmnonumber}{Theorem}

\usepackage{thmtools}
\usepackage{thm-restate}

\def\P{\mathbb{P}}

\def\S{\mathcal{S}}
\def\A{\mathcal{A}}

\def\O{\mathcal{O}}
\def\E{\mathbb{E}}

\newcounter{note}[section]

\DeclareMathOperator{\poly}{poly}

\renewcommand{\algorithmiccomment}[1]{\bgroup\hfill\footnotesize~#1\egroup}

\title{Learning Low Degree Hypergraphs}
\author{Eric Balkanski\footnote{Columbia University. \texttt{eb3224@columbia.edu}} \and Oussama Hanguir\footnote{Columbia University. \texttt{oh2204@columbia.edu}} \and Shatian Wang\footnote{Columbia University. \texttt{sw3219@columbia.edu}}}
\date{}

\begin{document}
\maketitle

\begin{abstract}%
 We study the problem of learning a hypergraph via edge detecting queries. In this problem, a learner queries subsets of vertices of a hidden hypergraph and observes whether these subsets contain an edge or not. 
 In general,  learning a  hypergraph  with $m$ edges of maximum size $d$ requires  $\Omega((2m/d)^{d/2})$ queries \citep{Angluin2006}. In this paper, we aim to identify families of hypergraphs that can be learned without suffering from a query complexity that grows exponentially in the size of the edges. 
    
    We show that hypermatchings and low-degree near-uniform hypergraphs with $n$ vertices are learnable with $\poly(n)$ queries. For learning hypermatchings (hypergraphs of maximum degree $ 1$), we give an $\O(\log^3 n)$-round  algorithm with $\O(n \log^5 n)$ queries. We complement this upper bound by showing that there are no algorithms with $\poly(n)$ queries that learn hypermatchings in $o(\log \log n)$ adaptive rounds.     For hypergraphs with maximum degree $\Delta$ and edge size ratio $\rho$, we give a non-adaptive algorithm with $\O((2n)^{\rho \Delta+1}\log^2 n)$ queries. To the best of our knowledge, these are the first algorithms  with $\poly(n, m)$ query complexity for learning non-trivial families of hypergraphs that have  a super-constant number of edges of super-constant size.\footnote{Accepted for presentation at the Conference on Learning Theory (COLT) 2022.}
\end{abstract}

\pagenumbering{gobble}
\clearpage
\newpage

\pagenumbering{arabic}

\section{Introduction}

Hypergraphs are a powerful tool in computational  chemistry and molecular biology where they are used to represent groups of chemicals and molecules that cause a reaction.
The chemicals and molecules that cause a reaction are however often unknown a priori, and learning such groups is a central problem of interest that has motivated a long line of work on hypergraph learning in the edge-detecting queries model, e.g., \citep{Torney1999, Angluin2006, angluin2004learning, abasi2014exact, Abasi2015, abasi2018error}. In this model, a learner queries subsets of vertices and, for each queried subset, observes whether it contains an edge or not. When the vertices represent chemicals, the learner queries groups of chemicals and observes a reaction if a subgroup reacts. The goal is to learn the edges with a small number of queries, i.e., a small number of  experiments.


An important application of learning chemical reaction networks is to learn effective drug combinations (drugs are vertices and effective drug combinations are hyperedges). In particular, there has recently been a lot of interest in finding drug combinations that reduce cancer cell viability. For example, as part of  AstraZeneca's recent DREAM Challenge, the effectiveness of    a large number of drug combinations was tested against different cancer  cell lines \citep{menden2019community}. Despite this interest, as noted by \citet{flobak2019high}, ``our knowledge about beneficial targeted drug combinations is still limited, partly due to the combinatorial complexity", especially since the time required to culture, maintain and test cell lines is significant. This combinatorial complexity motivates the need for novel computational methods that efficiently explore  drug combinations.

\vspace{.5cm}

One main obstacle to discovering effective combinations that involve $d > 2$ vertices is that $\Omega((2m/d)^{d/2})$ queries are required to learn hypergraphs with $m$ edges of maximum size $d$ \citep{angluin2004learning}.  Since there is no efficient algorithm for learning general hypergraphs with large maximum edge size, the main question is which families of hypergraphs can be efficiently learned.  

\begin{center}
\emph{Which families of hypergraphs can we learn with a number of queries \\ that does not grow exponentially in the size of the edges?}
\end{center}

\paragraph{Our results.} We develop the first algorithms with $\poly(n, m)$ queries for learning non-trivial families  of hypergraphs that have  a super-constant number of edges $m$ of super-constant size. The first family of such hypergraphs that we consider  are hypermatchings, i.e., hypergraphs such that every vertex is in at most one edge. Learning a hypermatching generalizes the problem of learning a matching studied in \citep{alon2004learning, beigel2001optimal}. In addition to their query complexity, we are also interested in the adaptive complexity of our algorithms, which is the number of adaptive rounds of queries required  when 
the algorithm can perform multiple queries in parallel in each round. Our first main result is an algorithm  with near-linear query complexity and poly-logarithmic adaptive complexity  for learning hypermatchings.


\begin{thmnonumber}
There is a  $\O(\log^3 n)$-adaptive algorithm with $\O(n \log^5 n)$ queries   that, for any hypermatching $M$ over $n$ vertices, learns $M$ with high probability.
\end{thmnonumber}

The adaptivity of the algorithm can be improved to $\O(\log n)$ at the expense of $\O(n^3 \log^3 n)$ queries.  We complement our algorithm by showing that there are no $o(\log \log n)$-adaptive algorithms that learn hypermatchings using $\poly(n)$ queries. 

\begin{thmnonumber}
There is no $o(\log \log n)$-adaptive algorithm with polynomial query complexity that learns an arbitrary hypermatching $M$ over $n$ vertices, even with small probability.
\end{thmnonumber}

Moving beyond hypermatchings, 
we  study hypergraphs whose vertices have bounded maximum degree $\Delta \geq 2$ ($\Delta = 1$ for hypermatchings). For such hypergraphs, the previously mentioned $\Omega((2m/d)^{d/2})$ lower bound on the number of queries needed by any algorithm implicitly also implies that $\Omega((m/\rho)^\rho)$ queries are required, even when $\Delta = 2$, where $\rho \geq 1$ is the ratio between the maximum and minimum edge sizes \citep{angluin2004learning}. Thus, an exponential dependence on this near-uniform edge sizes parameter $\rho$ is, in general, unavoidable. We give a non-adaptive algorithm with query complexity that depends on the maximum degree $\Delta$ and the near-uniform parameter $\rho$ of the hypergraph $H$ we wish to learn.

\begin{thmnonumber}
There is a  non-adaptive algorithm with $\O((2n)^{\rho \Delta +1} \log^2n)$ queries   that, for any $\rho$ near-uniform hypergraph $H$ with maximum degree $\Delta \geq 2$, learns $H$ with high probability.
\end{thmnonumber}
This query complexity is independent of the maximum size $d$ of the edges and is polynomial in the number of vertices $n$ when the maximum degree $\Delta$ and the edge size ratio $\rho$ are constant. Our learning algorithms rely on novel constructions of collections of queries that satisfy a simple property that we  call unique-edge covering, which is a general approach for constructing learning algorithms with a query complexity that does not grow exponentially in the size of the edges. We believe that  an interesting direction for future work is to identify, in addition to hypermatchings, other relevant families of hypergraphs that are learnable with $\poly(n,m)$ queries, even when both the maximum edge size $d$ and the edge size ratio $\rho$ are non-constant.

\paragraph{Technical overview.} Previous work on learning a hypergraph $H = (V, E)$ relies on constructing  a collection of sets of vertices called an independent covering family \citep{Angluin2006} or a cover-free family \citep{Abasi2015,gao2006construction, hwang2006pooling,abasi2014exact}. These families aim to identify \emph{non-edges}, which are sets of vertices  $T \subseteq V$
that do not  contain an edge $e \in E$. More precisely, both of these families are a collection $\S$ of sets of vertices such that every non-edge set   $T$ of size $d$ is contained in at least one set $S \in \S$ that does not contain an edge. These families have a minimum size that is exponential in $d$ and these approaches require a number of queries that is at least the size of these families. In contrast to these existing approaches that focus on non-edges, we construct a collection $\S$ of sets of vertices such that every \emph{edge} $e$ is contained in at least one set $S \in \S$ that does not contain any other edge of $H$. We call such a collection a \emph{unique-edge covering family}. Since the number of edges in a  hypergraph $H$ with maximum degree $\Delta$ is at most $\Delta n$, there exists unique-edge covering families  whose sizes depend on the  maximum degree $\Delta$ of $H$ instead of the maximum edge size $d$.

The algorithm for learning a hypermatching $M$ proceeds in phases. In each phase, we construct a unique-edge covering family of a subgraph of $M$ containing all edges of size that is in a specified range. This unique-edge covering family is constructed with i.i.d. sampled sets that contain each vertex, independently, with probability depending on the edge size range. The edge size range is widened at the end of each phase. One main challenge with hypermatchings is to obtain a near-linear query complexity. We do so by developing a subroutine which, given a set $S$ that covers a unique edge of size at most $s$, finds this unique edge with $\O(s \log |S|)$ queries. For the lower bound for hypermatchings, there is a simple construction that shows that there are no non-adaptive algorithms for learning hypermatchings with $\poly(n)$ queries. The technical part of interest is extending this construction and its analysis to hold for $\Omega(\log \log n)$ rounds. Finally, for low-degree near-uniform hypergraphs $H$, we construct a unique-edge covering family in a similar manner to those for hypermatchings ($\Delta = 1$). The main technical difficulty for hypergraphs with maximum degree $\Delta > 1$ is in analyzing the number of samples required to construct a unique-edge covering family, which is significantly more involved than for hypermatchings due to overlapping edges.

\paragraph{Related work.} The problem of learning a hidden hypergraph with edge-detecting queries was first introduced by \citet{Torney1999} for complex group testing, 
where this
problem is also known as exact learning from membership queries of a monotone DNF with at most $m$ monomials, where each monomial contains at most $d$ variables \citep{angluin1988queries,angluin2004learning,bshouty2018exact}.

Regarding hardness results, non-adaptive algorithms for learning hypergraphs with $m$ edges of size at most $d$ require $\Omega(N(m, d) \log n)$ queries where $N(m,d) = \frac{m+d}{\log \binom{m+d}{d}}\binom{m+d}{d}$ \citep{abasi2014exact,hwang2006pooling}. \citet{angluin2004learning} show that $\Omega((2m/d)^{d/2})$ queries are required by any (adaptive) algorithm  for learning general hypergraphs, and then extended this lower bound in \citep{Angluin2006} to $\Omega((2m/(\delta + 2))^{1+\delta/2})$ for  learning near-uniform hypergraphs with maximum edge size difference $\delta = d - \min_{e \in H}|e|$. 

Regarding adaptive algorithms, \citet{Angluin2006} give an algorithm that learns a hypergraph with $\O(2^{\O((1+\frac{\delta}{2})d)}\cdot m^{1+\frac{\delta}{2}} \cdot poly(\log n))$ queries in $\O((1+\delta)\min\{2^d(\log m+d)^2,(\log m+d)^3\})$ rounds with high probability. When $m < d$,  \citet{abasi2014exact} give a randomized algorithm that achieves a query complexity of $\O(dm \log n + (d/m)^{m-1+o(1)})$ and show an almost matching $\Omega(dm \log n + \left({d}/{m}\right)^{m-1})$ lower bound. When $m \geq d$,  they present a randomized algorithm that asks $(cm)^{d/2+0.75}+ dm \log n$  queries for some constant $c$, almost matching the   $\Omega((2m/d)^{d/2})$ lower bound of \citep{angluin2004learning}.  \citet{Chodoriwsky2015} develop an adaptive algorithm with $\O(m^d \log n)$ queries.  Adaptive algorithms with  $\O\left(md(\log_2{n} +(md)^d)\right)$ queries are constructed in  \cite{Dyachkov2016}. Regarding non-adaptive algorithms, \citet{chin2013non} and \citet{Abasi2015} give deterministic non-adaptive algorithms with an almost optimal  query complexity $N(m, d)^{1+o(1)}\log n$.  The problem of  learning a hypergraph when a fraction of the queries are incorrect is studied in \citet{chen2008upper} and \citet{abasi2018error}. We note that, to the best of our knowledge, all existing algorithms  require a number of queries that is exponential in $d$, or exponential in $m$ when $m < d$. We develop the first algorithms using $\poly(n)$ queries for learning non-trivial families  of hypergraphs that have arbitrarily large edge sizes and  a super-constant number of edges $m$.

Finally, learning a matching  has been studied in the context of graphs $(d=2)$ \citep{alon2004learning, beigel2001optimal}. In particular \citet{alon2004learning} provide a non-adaptive algorithm with $\O(n\log n)$ queries. To the best of our knowledge, there is no previous work on learning hypermatchings, which generalizes matchings to hypergraphs of maximum edge size $d > 2$. Our lower bound shows that, in contrast to algorithms for learning matchings, algorithms for learning hypermatchings require multiple adaptive rounds. Similar to the algorithm in \cite{alon2004learning} for learning matchings, our algorithm for learning hypermatchings has a near-linear query complexity.

\vspace{-3mm}

\section{Preliminaries}

\vspace{-1mm}

A hypergraph $H = (V,E)$ consists of a set of $n$ vertices $V$ and a set of $m$ edges $E \subseteq 2^V$. We abuse notation and write $e\in H$ for edges $e \in E$.
For the problem of learning a hypergraph $H$, the learner is given $V$ and  an edge-detecting oracle $Q_H$. Given a query $S \subseteq V$, the  oracle  answers $Q_H(S) = 1$ if $S$ contains an edge $e$, i.e. there exists $e \in H$ such that $e \subseteq S$; otherwise, $Q_H(S) = 0$. The goal is to learn the edge set $E$ using a small number of queries. When queries can be evaluated in parallel, we are interested in algorithms with low adaptivity. The \textit{adaptivity} of an algorithm is measured by the number of sequential rounds it makes where, in each round, queries are evaluated in parallel. An algorithm is non-adaptive if its adaptivity is $1$. 

The degree of a vertex $v$ is the number of edges $e \in H$ such that $v \in e$. The maximum degree $\Delta$ of $H$ is the maximum degree of a vertex $v \in V$. When $\Delta = 1$, every vertex is in at most one edge and $H$ is called a hypermatching, which we also denote by $M$. The rank $d$ of $H$ is the maximum size of an edge $e \in H$, i.e., $d := \max_{e \in H}|e|$. The edge size ratio of a hypergraph $H$ is  $d / \min_{e \in H} |e|$. A graph is $\rho$-near-uniform  and is uniform if $ d / \min_{e \in H} |e| \leq \rho$ and $ d / \min_{e \in H} |e| = 1$ respectively.

\vspace{-3mm}

\section{Learning Hypermatchings}
\label{sec:hypermatchings}
\vspace{-1mm}

In this section, we study the problem of learning hypermatchings, i.e., hypergraphs of maximum degree $\Delta = 1$. In Section \ref{subsection:learning_hypermatching}, we present an  algorithm that, with high probability, learns a hypermatching with $\O(n \log^5 n)$ queries in $\O(\log^3 n)$ rounds. The number of rounds can be improved to $\O(\log n)$, at the expense of an $\O(n^{3} \log^3 n)$ query complexity. In Section \ref{subsection:hardness_hypermatching}, we show that there is no $o(\log \log n)$-round algorithm that learns hypermatchings with $\poly(n)$ queries. Some proofs are deferred to Appendix \ref{appendix:learning_hypermatching} and \ref{appendix:hardness_hypermatching}.

\vspace{-1mm}
\subsection{Learning algorithm for hypermatchings}\label{subsection:learning_hypermatching}



A central concept we introduce for our learning algorithms is the definition of a unique-edge covering family of a hypergraph $H$, which is a collection of sets such that every edge $e \in H$ is contained in at least one set that does not contain any other edge.

\begin{definition}
A collection $\S \subseteq 2^V$ is a \textbf{unique-edge covering family} of $H$ if, for every $e \in H$, there exists $S \in \S$ s.t. $e$ is the unique edge contained in $S$, i.e. $e \subseteq S$ and $e' \not \subseteq S$ for all $e' \in H, e' \neq e$.
\end{definition}

\paragraph{Efficiently searching for the unique edge in a set.} We first show that, given a unique-edge covering family $\S$ of a hypermatching $M$, we can learn $M$. We observe that a set of vertices $S \subseteq V$ contains a unique edge if and only if $Q_M(S) = 1$ and $Q_M(S \setminus \{v\}) = 0$ for some $v \in S$ and that if it contains a unique edge, this edge is $e = \{v \in S: Q_M(S \setminus \{v\}) = 0\}$. This observation immediately leads to the following simple algorithm called $\textsc{FindEdgeParallel}$. 


\begin{algorithm}
	\caption{$\textsc{FindEdgeParallel(S)}$, returns $e$ if $S$ contains a unique edge $e$}
		\begin{algorithmic}[1]
	\INPUT{set $S \subseteq V$ of vertices}
	 \STATE{ \textbf{if} $Q_M(S) = 1$ and $\exists v \in S$ s.t. $Q_M(S \setminus \{v\}) = 0$  \textbf{then return} $\{v \in S: Q_M(S \setminus \{v\}) = 0\}$}
	 \STATE{\textbf{return} None}
	\end{algorithmic} 
	\label{alg:findedgeparallel}
\end{algorithm}	

The main lemma for \textsc{FindEdgeParallel} follows immediately from the above observation.
\begin{restatable}{rLem}{lemfuep}
\label{lem:fuep} For any hypermatching $M$ and set $S \subseteq V$,
\textsc{FindEdgeParallel} is a non-adaptive algorithm with $|S| + 1$ queries that, if $S$ contains a unique edge $e$, returns $e$  and None otherwise.
\end{restatable}

\begin{proof}
It is clear that $\textsc{FindEdgeParallel}(S)$ makes at most $|S| + 1$ queries. If $S$ does not contain an edge, then $Q_M(S) = 0$ and the algorithm returns None. If $S$ contains at least two edges, then there is no intersection between the edges because $M$ is a matching. Therefore, for every $v \in S$, $Q_M(S \setminus \{v\}) = 1$, and the algorithm returns None. The only case left is when $S$ contains a unique edge $e$. In this case $Q_M(S) = 1$ and $e = \{ v \in S, \ Q_M(S \setminus \{v\} = 0\}$. In this case, $\textsc{FindEdgeParallel}(S)$ return $e$.
\end{proof}

In order to obtain a near-linear query complexity for learning a hypermatching,  the query complexity of $\textsc{FindEdgeParallel}$ needs to be improved. Next, we describe an $\O(\log |S|)$-round algorithm, called \textsc{FindEdgeAdaptive}, which finds the unique edge $e$ in a set $S$ with $\O(s \log |S|)$ queries, assuming $|e| \leq s$. The algorithm recursively partitions  $S$ into two arbitrary sets $S_1$ and $S_2$ of equal size. If $Q_M(S_1) =  Q_M(S_2) = 1$, then $S$ contains at least two edges. If $Q_M(S_1) = 1$ and  $Q_M(S_2) = 0$ (or similarly $Q_M(S_1) = 0$ and  $Q_M(S_2) = 1$), then, assuming $S$ contains a single edge, this edge is contained in $S_1$ and we recurse on $S_1$.  The  most interesting case is if $Q_M(S_1) = Q_M(S_2) = 0$. Assuming $S$ contains a single edge $e$, this implies that both $S_1$ and $S_2$ contain vertices in $e$ and we recurse on both $S_1$ and $S_2$. When recursing on $S_1$, the set $S_2$ needs to be included in future queries  to find vertices in $e \cap S_1$ (and vice-versa when recursing on $S_2$). The algorithm thus also specifies an additional argument $T$ for the recursive calls. This argument is initially the empty set, in the case where $Q_M(S_1) = Q_M(S_2) = 0$ the set $S_2$ is added to $T$ when recursing on $S_1$ and vice-versa, and $T$ is included in all queries in the recursive calls. 

\begin{algorithm}
	\caption{$\textsc{FindEdgeAdaptive}(S, s)$, returns $e$ if $S$ contains a unique edge $e$ and $|e| \leq s$}
	\begin{algorithmic}[1]
	\INPUT{set $S \subseteq V$,  edge size $s$}
	\STATE{\textbf{if} $Q_M(S) = 0$ \textbf{then return}  None}
	\STATE{$\mathcal{S} \leftarrow \{(S, \emptyset)\}, \hat{e} \leftarrow \{\}$}
	\STATE{\textbf{while} $|\mathcal{S}| > 0$ \textbf{do}}
	\STATE{\quad \textbf{if} $|\mathcal{S}| > s$ \textbf{then return} None}
	\STATE{\quad $\mathcal{S}' \leftarrow \{\}$}
	\STATE{\quad \textbf{for} each $(S', T) \in \mathcal{S}$  \textbf{do} (in parallel)}
	\STATE{\quad \quad $S_1, S_2 \leftarrow $ partition $S'$ into two sets of size $|S'|/ 2$ }
	\STATE{\quad \quad \textbf{if} $|S'| = 1$ \textbf{then} add $v \in S'$ to $\hat{e}$}
	    \STATE{\quad \quad \textbf{else if} $Q_M(S_1 \cup T) = 1$ and $Q_M(S_2 \cup T) = 1$ \textbf{then return} None}
	    \STATE{\quad \quad \textbf{else if} $Q_M(S_1 \cup T) = 1$ and $Q_M(S_2 \cup T) = 0$ \textbf{then} add $(S_1, T)$ to $\mathcal{S}'$}
	    \STATE{\quad \quad \textbf{else if} $Q_M(S_1 \cup T) = 0$ and $Q_M(S_2 \cup T) = 1$ \textbf{then} add $(S_2, T)$ to $\mathcal{S}'$}	 
	    \STATE{\quad \quad \textbf{else if} $Q_M(S_1 \cup T) = 0$ and $Q_M(S_2 \cup T) = 0$ \textbf{then}  add $(S_1, S_2 \cup T), (S_2, S_1 \cup T)$ to $\mathcal{S}'$}
	  \STATE{\quad $\mathcal{S} \leftarrow \mathcal{S}'$}
	  \STATE{\textbf{if} $|\hat{e}| > s$, or $Q_M(\hat{e}) = 0,$ or $Q_M(\hat{e} \setminus \{v\}) = 1$ for some $v \in \hat{e}$ \textbf{then return}  None}
	 \RETURN $\hat{e}$
	\end{algorithmic} 
	\label{alg:findvertices}
\end{algorithm}

Let $\S^i = \{(S^i_1, T^i_1), \ldots, (S^i_\ell, T^i_\ell)\}$ be the $\S$ at the beginning of the $i$th iteration of the while loop of \textsc{FindEdgeAdaptive}$(S,s)$. In the case where $S$ contains a single edge $e$ and $|e| \leq s$, the algorithm maintains two invariants (formally stated in Lemma \ref{lem:fvconditions}; proof in Appendix \ref{appendix:learning_hypermatching}). The first is that, for all vertices $v \in e$ and all iterations $i$, either $v\in \hat{e}$ or there exists a unique $j$ such that $v$ is contained in set $S^i_j$. This invariant makes sure that every vertex $v \in e$ will eventually be added to $\hat{e}$. The second is that all sets  $S^i_j$ contain at least one vertex $v \in e$. This invariant explains why, in the base case $|S^i_j| = 1$, we add the single vertex in $S^i_j$ to the solution $\hat{e}$. Additionally, since by construction, $S^i_1, \ldots, S^i_l$ are disjoint, the second invariant also implies that when $|\S^i| > s$, we have that either set $S$ contains more than one edge or that it contains a single edge of size greater than $s$, in which case the algorithm returns None.  This stopping condition ensures that at most $2s$ queries are performed at each iteration of the while loop. 


\begin{restatable}{rLem}{lemfvconditions}
\label{lem:fvconditions}
Assume there is a unique edge $e \in M$ such that $e \subseteq S$ and that $|e| \leq s$, then, for every vertex $v \in e$ and at every iteration $i$ of the while loop of \textsc{FindEdgeAdaptive}$(S,s)$ we have that (1) either $v \in \hat{e}$ or $v \in S^i_j$ for a unique $j \in [\ell]$ and (2) $e \cap S^i_j \neq \emptyset$ for all $j \in [\ell]$.
\end{restatable}

We show that if $\textsc{FindEdgeAdaptive}(S, s)$ returns $\hat{e}$, then $\hat{e}$ is an edge. If there is a unique edge $e \subseteq S$ and $|e| \leq s$, then the algorithm returns $\hat{e} = e$.
 
\begin{restatable}{rLem}{lemfvcorrect}
\label{lem:fvcorrect}
For any $S \subseteq V$ and $s \in [n]$,
if there is a unique edge $e \in M$ such that $e \subseteq S$ and this edge is such that $|e| \leq s$, then $\textsc{FindEdgeAdaptive}(S, s)$ returns the edge $e$. Otherwise, it either returns None or an edge $e \in M$. $\textsc{FindEdgeAdaptive}(S, s)$ uses at most $2s \log_2 |S| + (s + 1)$ queries in at most  $\log_2 |S|$ rounds.
\end{restatable}

\begin{proof}
First, assume that there is a unique edge $e \in M$ such that $e \subseteq S$ and that $|e| \leq s$. Consider $v \in S$. If $v \in e$, then by Lemma~\ref{lem:fvconditions}, we have that at every iteration $i$ of the while loop, either $v \in \hat{e}$ or $v \in S^i_j$ for some $j \in \ell$. Since $|S^i_j| = |S^{i-1}_{j'}|/2$ for all $i, j, j'$, at iteration $i^{\star} = \log_2 n$, either $v \in \hat{e}$ or $v \in S^{i^{\star}}_j$ for some $j$ and $|S^{i^{\star}}_j| = 1$, in which case $v$ is then added to $e$ by the algorithm. Next, if $v \not \in e$, since $e \cap S^i_j \neq \emptyset$ for all $i, j$ by Lemma~\ref{lem:fvconditions}, $v$ is never added to $\hat{e}$. Thus, \textsc{FindEdgeAdaptive}$(S,s)$ returns the edge $e$


Assume now that $S$ does not contain any edges. In this case, every query is answered by zero, and \textsc{FindEdgeAdaptive}$(S,s)$ returns None.  Finally, assume that $S$  contains at least two edges. If \textsc{FindEdgeAdaptive}$(S,s)$ does not return None, then Step 14 ensures that the returned edge $\hat{e}$ has size at most $s$. Step 14 also ensures that $\hat{e}$ contains at least one edge. If $\hat{e}$ strictly contains an edge, then there is a vertex $v \in \hat{e}$ such that $Q_M(\hat{e} \setminus \{v\}) = 1$, in which case \textsc{FindEdgeAdaptive}$(S,s)$ would have returned None. Therefore, $\hat{e}$ is an edge of $M$ of size less than $s$.

We now show that \textsc{FindEdgeAdaptive}$(S,s)$ runs in $\log_2 |S|$ rounds and makes at most $2s\log_2 |S| + (s+1)$ queries. At every iteration of the while loop, \textsc{FindEdgeAdaptive}$(S,s)$ makes less than $2s$ queries (recall that if $|\S| > s$ then the algorithm returns None). We show that the number of iterations of the while loop is less than $\log_2 |S|.$ At every iteration of the while loop, for every couple $(S,T)$ in $\S$, the size of $S$ is divided by 2. After $\log_2 |S|$ iterations, either $\S$ is empty or for every couple $(S,T)$ in $\S$ we have $|S| = 1$. This guarantees that no sets are added to $\S'$. Therefore, the number of iterations of the while loop is less than $\log_2 |S|$. This also shows that the adaptive complexity of \textsc{FindEdgeAdaptive}$(S,s)$ is less than $\log_2 |S|$. After the while loop, we make at most $s+1$ queries
and the total number of queries is less than $2s \log |S| + s +1$.
\end{proof}

\paragraph{Constructing a unique-edge covering family.} Next, we give an algorithm called \textsc{FindDisjointEdges}  that, for any hypermatching $M$, constructs a unique-edge covering family $\S$ of the $\alpha$-near-uniform hypermatching $M_{s, \alpha, V'}$. $M_{s, \alpha, V'}$ is defined as the subgraph of $M$ over edges of size between $s/\alpha$ and $s$ and over vertices in $V'$ that are not in an edge of size less than $s/\alpha$. The  unique-edge covering family $\S$ is constructed with i.i.d. samples that contain each vertex in $V'$, independently, with probability $n^{-\alpha/s}$. The algorithm then calls a \textsc{FindEdge} subroutine, which is either \textsc{FindEdgeParallel} or \textsc{FindEdgeAdaptive}, on samples  that contain at least one edge.


\begin{algorithm}
	\caption{$\textsc{FindDisjointEdges}(s, \alpha, V', \textsc{FindEdge})$, returns edges of size between $\frac{s}{\alpha}$ and $ s$}
	\begin{algorithmic}[1]
	\INPUT{edge size $s$, parameter $\alpha \geq 1$, set of vertices $V' \subseteq V$, subroutine $\textsc{FindEdge}$}
	\STATE{$\hat{M} \leftarrow \emptyset$}
	\STATE{\textbf{for} $i = 1, \ldots, \ n^{\alpha}\log^2{n}$ \textbf{do} (in parallel)}
	    \STATE{\quad $S_i \leftarrow $ set of independently sampled vertices from $V'$, each with probability $n^{-\alpha/s}$.}
	    \STATE{\quad  \textbf{if} $Q_M(S_i) = 1$ \textbf{then}}
	\STATE{\quad \quad $e \leftarrow \textsc{FindEdge}(S_i,s)$}
	\STATE{\quad  \quad \textbf{if} $e$ is not None \textbf{then} add $e$ to $\hat{M}$}
	\RETURN $\hat{M}$.
	\end{algorithmic} 
	\label{alg:adaptive_new}
\end{algorithm}

We show that $n^{\alpha} \log^2 n$ such samples suffice to construct a unique-edge covering family $\S$ of the $\alpha$-near-uniform hypermatching $M_{s,\alpha,V'}$.

\begin{restatable}{rLem}{lemmafindedges} 
\label{lemma:find_disjoint_edjes}
Assume $s/\alpha \geq 2$, $\alpha \geq 1$, that the $\textsc{FindEdge}$ subroutine uses at most $q$ queries in at most $r$ rounds, and let $\S = \{S_i\}_i$ be the $n^{\alpha}\log^2 {n}$ samples, then $\textsc{FindDisjointEdges}$ is an $(r+1)$-adaptive algorithm with $n^{\alpha} \log^2 n + 2\alpha qn^{\alpha}\log^2n/s$ queries such that, with probability $1 - \O(n^{-\log n})$, $\S$ is a unique-edge covering family of the hypermatching $M_{s, \alpha, V'} = \{e \in M : e \subseteq V', s/\alpha \leq |e| \leq s \}$ and thus $\hat M = M_{s, \alpha, V'}$.
\end{restatable}


\begin{proof}
We know that a sample $S_i$ contains at least an edge if and only if $Q_M(S_i) = 1$. We also know from Lemma \ref{lem:fuep} and Lemma \ref{lem:fvcorrect} that if $S_i$ contains a unique edge $e$ of size less than $s$, then the subroutine $\textsc{FindEdge}(S_i, s )$ will return $e$. Therefore, to learn all edges of size in $\left[s/\alpha, s \right]$, we only need to make sure that each such edge appears in at least one sample that does not contain any other edges, in other words, that $\S$ is a unique-edge covering family for $M_{s,\alpha,V'} = \{ e \in M: e \subseteq V', \ s/\alpha \leq |e| \leq  s$\}. Below, we show that it is the case w.p. $1-O(n^{-\log n}) = 1 -e^{-\Omega(\log^2 n)}$.

We first show that with high probability, each edge $e$ of size $|e| \in \left[s/\alpha, s \right]$ is contained in at least $\log^2 n/2$ sample $S_i$'s. We use $X_e$ to denote the number of samples containing $e$, then we have
$\E(X_e) = n^{\alpha} \log^2 n \cdot n^{-\frac{\alpha|e|}{s}} \geq n^{\alpha} \log^2 n \cdot n^{-\alpha} = \log^2 n.$ By Chernoff bound, we have
$P(X_e \leq \log^2 n/2) \leq n^{-\log n/8}.$
As there are at most $\alpha n/s$ edges of size between $s/\alpha$ and $s$, by a union bound, $P(\exists e \in M \text{ s.t. } |e|\in \left[s/\alpha, s \right], X_e \leq \log^2 n/2) \leq \frac{\alpha n}{s}n^{-\log n/8} \leq n^{-(\log n/8-1)}.$ and we get $P(\forall e \in M \text{ s.t. } |e|\in \left[s/\alpha, s \right], X_e \geq \log^2 n/2) \geq 1 - n^{-(\log n/8-1)}.$ From now on we condition on the event that all edges $e$ whose size is between $s/\alpha$ and $s$ are included in at least $\log^2 n/2$ samples. We show below that given $e \subseteq S_i$ for a sample $S_i$, the probability that there exists another edge of size at least $s$ in $S_i$ is upper bounded by $\alpha/s$. Recall that $M$ is the hidden matching we want to learn. We abuse notation and use $e' \in M \cap S_i$ to denote that an edge $e'$ is both in the matching $M$ and included in the sample $S_i$. We have
\begin{align*}
\P(\exists e' \in M \cap S_i, e' \neq e \ | \ e \subseteq S_i) & \leq  \sum_{e'\in M, e' \neq e} \P(e' \subseteq S_i \ | \ e \subseteq S_i) \\
& =  \sum_{e'\in M, e' \neq e} \P(e' \subseteq S_i)
\leq  \frac{\alpha n}{s} \cdot (n^{-\frac{\alpha}{s}})^{\frac{s}{\alpha}}  = \frac{\alpha}{s},
\end{align*}
where the first inequality uses a union bound, the first equality is due to the fact that $M$ is a matching and thus $e\cap e' = \emptyset$ and that vertices are sampled independently, and the second inequality follows because the total number of remaining edges is upper bounded by $\alpha n/s$ and each edge is in $S_i$ with probability at most $(n^{-\frac{\alpha}{s}})^{\frac{s}{\alpha}}$. As each edge $e$ with size between $s/\alpha$ and $s$ is contained in at least $\log^2 n/2$ samples, we have 
\begin{align*}
\P(\forall S_i \text{ s.t. } e \subseteq S_i, \exists \ e' \in M \cap S_i, e' \neq e) \leq (\frac{s}{\alpha})^{-\log^2 n /2} \leq n^{-\log n/4},
\end{align*}
where the last inequality is since $s/\alpha \geq 2$. By another union bound on the  edges of size between $s/\alpha$ and $ s $, $
\P(\exists e \text{ s.t. } |e| \in \left[s/\alpha, s\right], \forall S_i \text{ s.t. } e \subseteq S_i, \exists e' \in M \cap S_i, e' \neq e) \leq n^{-(\log n/2-1)}.$
We can thus conclude that with probability at least $1-O(n^{-\log n}) = 1 - e^{-\Omega(\log^2 n)}$, for all $e \in M$ with size between $s/\alpha$ and $s$, there exists at least one sample $S_i$ that contains $e$ but no other remaining edges. By Lemma \ref{lem:fuep} and Lemma \ref{lem:fvcorrect}, $e$ is returned by $\textsc{FindEdge}(S_i,s)$, and $e$ is added to the matching $\hat{M}$ that is returned by $\textsc{FindDisjointEdges}(s,\rho,V')$.

Next, we show that $\textsc{FindDisjointEdges}$ is an $(r+1)$-adaptive algorithm that requires $n^{\alpha} \log^2 n +2\alpha q n^{\alpha}
\log^2 n/s$ queries. We first observe that $\textsc{FindDisjointEdges}$ makes only parallel calls to $\textsc{FindEdges}$ (after verifying that $Q_M(S_i) = 1$). Therefore, since $\textsc{FindEdges}$ runs in at most $r$ rounds, we get that $\textsc{FindDisjointEdges}$ runs in at most $r+1$ rounds. To prove a bound on the number of queries, we first argue using a Chernoff bound that with probability $1 - e^{-\Omega(\log^2 n)}$, every edge $e$ of size at least $s/\alpha$ is in at most $2n^{\alpha - 1}\log^2 n$ samples $S_i$. Therefore there are at most $\frac{2\alpha n^{\alpha}}{s} \log^2n$ samples $S_i$ such that $Q_M(S_i) = 1$. For every one of these samples, we call $\textsc{FindEdges}(S_i,s)$, which makes $q$ queries by assumption. Therefore, with high probability $1 - e^{-\Omega(\log^2 n)}$, $\textsc{FindDisjointEdges}$ makes $\O(n^{\alpha} \log^2 n + \alpha q  \frac{2n^{\alpha}}{s} \log^2n)$ queries.
\end{proof}

\paragraph{The main algorithm for hypermatchings.} The main algorithm first  finds edges of size $1$ by brute force with $\textsc{FindSingletons}$. It then iteratively learns the edges of size between $s/\alpha$ and $ s$ by calling $\textsc{FindDisjointEdges}(s,\alpha, V')$, where $V'$ is the set of vertices that are not in edges learned in previous iterations. At the end of each iteration, the algorithm increases $s$. We obtain the main theorem for learning hypermatchings, the proof of which is given in Appendix \ref{appendix:learning_hypermatching}.

\begin{algorithm}
	\caption{\textsc{FindMatching}$(\alpha, \textsc{FindEdge})$, learns a hypermatching.}
	\begin{algorithmic}[1]
	\INPUT{parameter $\alpha \geq 1$, subroutine $\textsc{FindEdge}$}
	\STATE{$\hat{M} \leftarrow \textsc{FindSingletons}, s \leftarrow 2\alpha$}
	\STATE{\textbf{while} $s < n$ \textbf{do}}
	   \STATE{\quad $V' \leftarrow V \setminus \{v : v \in e \text{ for some } e \in \hat{M}\}$}
	    \STATE{\quad $\hat{M} \leftarrow \hat{M} \cup \textsc{FindDisjointEdges}(s,\alpha, V', \textsc{FindEdge})$}
	    \STATE{\quad $s \leftarrow \lfloor   \alpha s \rfloor + 1$}
	\RETURN $\hat{M}$.
	\end{algorithmic} 
	\label{alg:adaptive_new_main}
\end{algorithm}

\vspace{-2mm}

\begin{restatable}{rThm}{thmlearninghypermatchingnew}\label{thm:learning_hypermatching_new} Let $M$
be a hypermatching, then $\textsc{FindMatching}$ learns $M$ w.h.p either
 in $\O(\log^3 n)$ rounds using $\O(n \log^5 n)$ queries,  with $\alpha = 1 / (1 - 1/(2\log n))$ and \textsc{FindEdgeAdaptive}, or
 in $\O(\log n)$ rounds using $\O(n^{3} \log^3 n)$ queries, with $\alpha = 2$ and \textsc{FindEdgeParallel}.\end{restatable}




\vspace{-6mm}
\subsection{Hardness of learning hypermatchings}\label{subsection:hardness_hypermatching}
\vspace{-1mm}
In this section, we show that the adaptive complexity of learning a hypermatching with $\poly(n)$ queries is $\Omega(\log\log n)$. This result is in sharp contrast to matchings (where edges are of size $d = 2$) for which there exists a non-adaptive learning algorithm
with $\O(n \log n)$ queries \citep{alon2004learning}. 
\vspace{-5mm}
\paragraph{Warm-up: hardness for non-adaptive learning.} As a warm-up, we present a simple construction of a family of hypermatchings for which there is no non-adaptive learning algorithm with $\poly(n)$ queries. Each hypermatching $M_P$ in this family depends on a partition $P = (P_1, P_2, P_3)$ of the vertex set $V$ into three parts $P_1, P_2, P_3$ where $|P_1| = n - (\sqrt{n} + 1)$, $|P_2| = 1$, and $|P_3| = \sqrt{n}$. $M_P$ contains $(n - (\sqrt{n} + 1))/2$ edges of size $2$ which form a perfect matching over $P_1$ and one edge of size $\sqrt{n}$ which contains all the vertices in $P_3$. The only vertex in $P_2$ is not contained in any edges in $M_P$. The main idea of the construction is that, after one round of $\poly(n)$ queries,  vertices in $P_2$ and $P_3$ are indistinguishable to the learning algorithm. However, the algorithm needs to distinguish vertices in $P_3$ from the vertex in $P_2$ to learn the edge $P_3$.



\vspace{-1.5mm}
\begin{restatable}{rThm}{thmhardnessone}\label{thm:hardness}
There is no non-adaptive algorithm with $\poly(n)$ query complexity that can learn every non-uniform hypermatching with probability  $\omega(1/\sqrt{n})$.
\end{restatable}
\vspace{-1.5mm}
\begin{proof}
 Let $\mathcal{P}_3$ be the set of all possible partitions $(P_1,P_2,P_3)$ such that $|P_1| = n - (\sqrt{n} + 1)$, $|P_2| = 1$, and $|P_3| = \sqrt{n}$. We consider  a uniformly random partition $P$ from $\mathcal{P}_3$, a matching $M_P$ and a non-adaptive algorithm $\A$ that  asks a collection $\S$ of $\poly(n)$ non-adaptive queries.  The main lemma (Lemma \ref{lemma:non_adaptive_hardness_1})  shows that, with high probability, for all queries $S \in \S$, the answer $Q_{M_P}(S)$ is independent of which $\sqrt{n}$ vertices in $P_2 \cup P_3$ constitute the edge $P_3$. The analysis of this lemma consists of two cases. If the query $S$ is small ($|S| < (n + \sqrt{n}+1)/2$), then we show that  $S$ contains less than  $\sqrt{n}$ vertices from $P_2 \cup P_3$, w.h.p. over the randomization of $P$, by the Chernoff bound. Therefore $S$ does not contain $P_3$ and $Q_{M_P}(S)$ is independent of the partition of $P_2 \cup P_3$ into $P_2$ and $P_3$. If $S$ is large ($|S| \geq (n + \sqrt{n}+1)/2$), then $S$ contains an edge of size two from $P_1$. Thus $Q_{M_P}(S) = 1$ and $Q_{M_P}(S)$ is independent of the partition of $P_2 \cup P_3$ into $P_2$ and $P_3$.
 
In Lemma \ref{lemma:non_adaptive_hardness_2}, we  argue that if all queries $Q_{M_P}(S)$ of $\A$ are independent of the partition of $P_2 \cup P_3$ into $P_2$ and $P_3$, then, with high probability, the matching returned by this algorithm is not equal to $M_P$. By the probabilistic method, there is a matching $M_P$ with $P \in \mathcal{P}_3$ that $\A$ does not successfully learn with probability $1 - \O(1/\sqrt{n})$.
\end{proof}

\vspace{-2mm}

\paragraph{Hardness of learning in $o(\log \log n)$ rounds.} The main technical challenge  is to generalize the construction and the analysis of the hard family of hypergraphs for non-adaptive learning algorithms to a construction and an analysis which holds over  $\Omega(\log \log n)$ rounds. In this construction, each hypermatching $M_P$  depends on a partition $P = (P_0, P_1, \ldots, P_R)$ of the vertex set $V$ into $R+1$ parts. For each $i \in \{0,\ldots R\}$, $M_P$ contains $|P_i|/d_i$ edges of size $d_i$  which form a perfect matching over $P_i$. We set the sizes such that $d_0 = 3$, $d_{i+1} = 3\log^{2} n \cdot d_i^2$,  and  $|P_{i}| = 3\log^{2} n \cdot d_i^2$. 

The main idea of the construction is that after $i$ rounds of queries, vertices in $\cup_{j = i+1}^R P_i$ are indistinguishable to any algorithm. However, the algorithm needs to distinguish vertices in $P_{j}$ from vertices in $P_{j'}$, for all $j \neq j'$, to learn $M_P$. Informally, since the edges in $P_{j'}$ have a larger size than edges in $P_{j}$ for $j' > j$, an algorithm can learn only one part $P_j$ at each round of queries, the one with edges of minimum size among all parts that have not yet been learned in previous rounds.

\vspace{-2mm}

\begin{restatable}{rThm}{thmhardnesslog}
\label{thm:hardnesslog}
There is no $(\log\log n - 3)$-adaptive algorithm with $poly(n)$ query complexity which can learn every non-uniform hypermatching with  probability $e^{-o(\log^2 n)}$.
\end{restatable}
\vspace{-2mm}
\begin{proof} [Proof Sketch, full proof in Appendix~\ref{appendix:learning_hypermatching}.] Let $\mathcal{P}_R$ be the set of all  partitions $(P_0, \ldots ,P_R)$ such that $|P_i| =  3\log^{2} n \cdot d_i^2$, and let $P$ be  a uniformly random partition from $\mathcal{P}_R$, and  $\A$ be an  algorithm with $\poly(n)$  queries. In Lemma~\ref{lemma:induction_loglogn_hardness}, we show that if vertices in $P_{i}, P_{i+1},  \ldots, P_R$ are indistinguishable to $\A$ at the beginning of round $i$ of queries, then vertices in $P_{i+1},  \ldots, P_R$ are indistinguishable at the end of round $i$.

The proof of Lemma~\ref{lemma:induction_loglogn_hardness} consists of two parts. First, Lemma \ref{lemma:small_query} shows that if $S$ is small ($|S| \leq (1-1/d_i)\sum_{j \geq i} |P_j|)$ and is independent of partition of $\cup_{j = i}^R P_i$ into $P_i, \ldots, P_R$, then w.h.p., for every $j \geq i+1$, $S$ does not contain any edge contained in $P_j$. Second, Lemma \ref{lemma:large_query} shows that if a query $S$ has a large size ($|S| \geq (1-1/d_i)\sum_{j \geq i} |P_j|)$ and is independent of partition of $\cup_{j = i}^R P_i$ into $P_i, \ldots, P_R$, then w.h.p. $S$ contains at least one edge from the matching on $P_i$ which implies $Q_{M_P}(S) = 1$. Proving Lemma \ref{lemma:large_query} is quite involved. Because the edges of a matching are not independent from each other, a Chernoff bound analysis is not enough to show that the probability that a query does not contain any edge from $P_i$ is small. We therefore provide an alternative method to bound this probability (Lemma \ref{claim:uniform_edges}).
In the end, in both cases, we show that $Q_{M_P}(S)$ is independent of the partition of $\cup_{j = i+1}^R P_i$ into $P_{i+1}, \ldots, P_R$. Finally, we conclude the proof of Theorem~\ref{thm:hardnesslog} similarly as for Theorem~\ref{thm:hardness}.
\end{proof}

\vspace{-5mm}

\section{Learning  Low-Degree Near-Uniform Hypergraphs}\label{sec:bounded_degree}
\vspace{-2mm}

In this section, we give an algorithm for learning low-degree near-uniform hypergraphs. The algorithm is non-adaptive and has  $\O\left((2n)^{\rho \Delta + 1}\log^2 n)\right)$ query complexity for learning a hypergraph $H(V,E)$ of maximum degree $\Delta$ and edge size ratio $\rho$. It generalizes $\textsc{FindDisjointEdges}$ and  constructs a unique-edge covering family, but its analysis is completely different, and significantly more challenging, due to overlapping edges. Full proofs are deferred to Appendix \ref{appendix:hypergraphs}.


\paragraph{Near-uniformity is necessary.} The $\O(n \log^5 n)$ query complexity from the previous section  for $\Delta = 1$ holds for any hypermatching $M$, even those with edge size ratio $\rho = \Omega(n)$. In contrast, for general $\Delta$, we obtain an $\O(n^{\rho \Delta + 1}\log^2 n)$ query complexity. \citet{Angluin2006} show that $\Omega((2m/d)^{d/2})$  queries are required to, even fully adaptively, learn hypergraphs of maximum edge size $d$. Their hardness result holds for a family of hypergraphs of maximum degree $\Delta = 2$ and edges of size $2$ or $d$, so with $\rho = d/2$. Thus, it implies that $\Omega((m/\rho)^{\rho})$ queries are required to learn hypergraphs  even when $\Delta = 2$, i.e., an exponential dependence on $\rho$ is required.



\paragraph{The algorithm.}  \textsc{FindLowDegreeEdges}  constructs a unique-edge covering family of the hypergraph $H$ similarly to \textsc{FindDisjointEdges}, which requires a larger  number of samples  when $\Delta > 1$. In addition, steps $5$-$6$ are needed to ensure that intersections of edges are not falsely identified as edges since, when $S_i$ contains two or more edges that all overlap in $S'_i \subseteq S_i$, we have that $S'_i = \{v \in S_i: Q_H(S_i \setminus \{v\}) = 0\}$. 
    



\begin{algorithm}[H]
	\caption{\textsc{FindLowDegreeEdges}, learns a $\rho$-near-uniform hypergraph with max degree $\Delta$}
	\label{alg:1round-boundeddegree-nu}
	\begin{algorithmic}[1]
	\INPUT{edge size ratio $\rho$, maximum edge size $d$, maximum degree $\Delta \geq 2$}
	\STATE{$\hat{H} \leftarrow \emptyset$}
	\STATE{\textbf{for} $i = 1, \ldots, \ (2n)^{\rho\Delta}\log^2 n $ \textbf{do} (in parallel)}
	    \STATE{ \quad $S_i \leftarrow $ set of independently sampled vertices, each with probability $p = (2n)^{-\frac{\rho}{d}}$.}
	\STATE{\quad \textbf{if}  $Q_H(S_i) = 1$ and $\exists v$ s.t. $Q_H(S_i \setminus \{v\}) = 0$ \textbf{then} add $\{v \in S_i: Q_H(S_i \setminus \{v\}) = 0\}$ to $\hat{H}$}

	\STATE{\textbf{for} $e \in \hat{H}$ \textbf{do} (in parallel)}
	    \STATE{\quad  \textbf{if} $\exists \ e' \in \hat{H}$ s.t. $e \subset e'$ \textbf{then} remove $e$ from $\hat{H}$}
	  \RETURN{$\hat{H}$}
	\end{algorithmic} 
\end{algorithm}

Recall that for hypermatchings, 
\textsc{FindDisjointEdges} is used as a subroutine in an adaptive procedure that learns edges of any size. For hypergraphs, however, we cannot  call \textsc{FindLowDegreeEdges} iteratively and remove vertices in learned edges after each call because a large edge could overlap with a small edge.

\paragraph{The analysis.} The main technical challenge in this section is to analyze the sample complexity required to obtain w.h.p. a unique-edge covering family. The main lemma bounds the probability that,  conditioned on a sample $S$ containing an edge $e$, this sample contains another edge $e'$.


\begin{lemma} \label{lemma:failure-prob-bn}
If a sample set $S$ is constructed by independently sampling each vertex w.p. $p = (2n)^{-\frac{\rho}{d}}$ from a hypergraph $H$ with maximum degree $\Delta \geq 2$, maximum edge size $d$, edge size ratio $\rho$,  $n \geq 100$, then for any edge $e \in H$,
$\P(\exists e' \subseteq S, e' \in H, e' \neq e \ | \ e \subseteq S) \leq 1-\left(\frac{\log n}{2n}\right)^{(\Delta - 1)\rho}.$
\end{lemma}
\textbf{Proof.} [Proof Sketch, full proof in Appendix~\ref{appx:lemma_failure_prob}.] Given an edge $e \in S$,
the proof of Lemma~\ref{lemma:failure-prob-bn} relies on analyzing  the following two mathematical programs, whose optimal values are used to bound $\P(\exists e' \subseteq S, e' \in H, e' \neq e \ | \ e \subseteq S)$.

\noindent\begin{minipage}{.5\linewidth}
\begin{alignat*}{2}
 & \min\limits_{a_{ij}} \hspace{.1cm} && \prod\limits_{j=d/\rho}^{d}\prod\limits_{i=0}^{j-1} (1-p^{j-i})^{a_{ij}} \\
 &  \hspace{.1cm} \text{s.t.} && \sum_{j = d/\rho}^{d}\sum_{i = 0}^{j-1} i\cdot  a_{ij} \leq (\Delta -1)d\\
& && \sum_{j = d/\rho}^{d}\sum_{i = 0}^{j-1} j\cdot  a_{ij} \leq \Delta n\\
& && \hspace{.1cm} a_{ij} \geq 0.
\end{alignat*}
\end{minipage}%
\begin{minipage}{.5\linewidth}
  \begin{alignat*}{2}
& \max\limits_{a_{ij}} \hspace{.1cm} && \sum_{j = d/\rho}^{d}\sum_{i = 0}^{j-1} a_{ij} \cdot p^{j-i} \\
& \hspace{.1cm} \text{s.t.}  && \sum_{j = d/\rho}^{d}\sum_{i = 0}^{j-1} i\cdot  a_{ij} \leq (\Delta -1)d\\
& && \sum_{j = d/\rho}^{d}\sum_{i = 0}^{j-1} j\cdot  a_{ij} \leq \Delta n\\
& && \hspace{.1cm} a_{ij} \geq 0. 
\end{alignat*}
\end{minipage}



The variables $a_{ij}$ denote the number of edges of size $j$ intersecting $e$ in $i$ vertices and we assume $|e| = d$. The two programs have identical constraints. The first constraint is that the sum, over edges $e' \neq e$, of the size of the intersection of $e$ and $e'$ has to be at most $(\Delta - 1) d$ since $e$ has size $d$ and each vertex has degree at most $\Delta$. The second constraint is that the sum, over edges $e' \neq e$, of the size of $e'$ has to be at most $\Delta n$ since there are $n$ vertices of degree at most $\Delta.$ 

For the objectives, we note that $p^{j-i}$ is the probability that a fixed collection of $j-i$ vertices is in a sample. Thus, $1 - p^{j-i}$ is the probability that an edge $e'$ of size $j$ intersecting $e$ in $i$ vertices is not contained in a sample $S$, conditioned on $e \subseteq S$. If events $e_1 \subseteq S | e \subseteq S$ and $e_2 \subseteq S| e \subseteq S$ were independent for all $e_1, e_2 \neq e$ (which is not the case), the probability that there is no other edge in $S$, conditioned on $e \in S$ would be equal to the objective of the left program. The objective of the right program is an upper bound on $\P(\exists e' \subseteq S, e' \in H, e' \neq e \ | \ e \subseteq S)$ by a union bound over all edges $e' \neq e$. We denote the optimal value to the left and right programs by   $f^\bullet(\Delta, p, d, \rho)$ and $LP^\bullet(\Delta, p, d, \rho)$. The main steps of the proof are





\begin{enumerate}
\item In Lemma~\ref{lemma:success_probability_bounded}, we show that the unfavorable event is more likely to happen with the edge independence assumption:
\begin{equation}\label{eq:failure-prob-bn-1}\P(\exists e' \subseteq S, e' \in H, e' \neq e \ | \ e \subseteq S) \leq 1-f^\bullet(\Delta,p,d, \rho)\end{equation}
by using  induction and Bayes rule to formalize the intuition that the events $\{e' \not \subseteq S\}$ for $e' \neq e$ are positively correlated because edges could share common vertices.  
\item In Lemma~\ref{lemma:exp-nb}, we show that for $\Delta \geq 2$,
\begin{equation}\label{eq:failure-prob-bn-2}f^\bullet(\Delta,p,d, \rho) \geq f^\bullet(2,p,d/\rho, 1)^{\rho(\Delta - 1)}\end{equation}
by deriving a closed-form optimal solution of $f^\bullet(\Delta,p,d, \rho)$, which maximizes $a_{ij}$ for $i \in \{0, d/\rho -1\}, j = d/\rho$ and sets the rest $a_{ij} = 0$.
\item In Lemma~\ref{lemma:f-lp-2}, we show that
\begin{equation}\label{eq:failure-prob-bn-3}f^\bullet(2,p,d/\rho, 1) \geq 1 - LP^\bullet(2,p,d/\rho, 1).\end{equation}
\item In Lemma~\ref{lemma:lp2ub-nb}, we show that 
\begin{equation}\label{eq:failure-prob-bn-4}1-LP^\bullet(2,p,d/\rho, 1) \geq \frac{\log n}{2n}.\end{equation}
by deriving a closed-form solution to $LP^\bullet(2, p, d/\rho, 1)$ and showing that  it can be upper bounded by $\frac{d/\rho}{d/\rho-1}n^{-\rho/{d}}+\frac{1}{d/\rho}$ and that $\frac{d/\rho}{d/\rho-1}n^{-\rho/{d}}+\frac{1}{d/\rho}\leq 1 - \frac{\log n}{2n}$ for $n \geq 100$.
\item Combining Eqs.\eqref{eq:failure-prob-bn-1}, \eqref{eq:failure-prob-bn-2}, \eqref{eq:failure-prob-bn-3}, \eqref{eq:failure-prob-bn-4}, we get the desired bound. \hfill  
$\blacksquare$
\end{enumerate}

We are now ready to present the main result for this section (see Appendix \ref{appx:proof_main_thm} for the full proof). 

\begin{restatable}{rThm}{thmmain}
\label{thm:main} For any $\rho$-near-uniform hypergraph $H$ with maximum degree $\Delta \geq 2$, maximum edge size $d$,  and number of vertices $n \geq 100$, \textsc{FindLowDegreeEdges} correctly learns $H$ with probability  $1-o(1)$. Furthermore, it is non-adaptive and makes at most $\O((2n)^{\rho\Delta+1}\log^2 n)$ queries.
\end{restatable}

We note that if the exact parameters $\rho$, $d$, and $\Delta$ are unknown and only  upper bounds $\bar{\rho} \geq \rho$, $\bar{d} \geq d$, $\bar{\Delta} \geq \Delta$ are available, then, if $\bar{\rho}$ is large enough so that  $\bar{d}/\bar{\rho} \leq \min_{e \in H} |e|$, we have that \textsc{FindLowDegreeEdges} with inputs $\bar{\rho}, \bar{d}, \bar{\Delta}$  also learns $H$ with probability $1-o(1)$. 

\newpage

\section*{Acknowledgements}We thank Meghan Pantalia and Matthew Ulgherait for helpful discussions on finding drug combinations that reduce cancer cell viability.

\bibliographystyle{plainnat}
\bibliography{bibliography}

\begin{thebibliography}{18}
\providecommand{\natexlab}[1]{#1}
\providecommand{\url}[1]{\texttt{#1}}
\expandafter\ifx\csname urlstyle\endcsname\relax
  \providecommand{\doi}[1]{doi: #1}\else
  \providecommand{\doi}{doi: \begingroup \urlstyle{rm}\Url}\fi

\bibitem[Abasi(2018)]{abasi2018error}
Hasan Abasi.
\newblock Error-tolerant non-adaptive learning of a hidden hypergraph.
\newblock In \emph{43rd International Symposium on Mathematical Foundations of Computer Science (MFCS 2018)}. Schloss Dagstuhl-Leibniz-Zentrum fuer Informatik, 2018.

\bibitem[Abasi et~al.(2014)Abasi, Bshouty, and Mazzawi]{abasi2014exact}
Hasan Abasi, Nader~H Bshouty, and Hanna Mazzawi.
\newblock On exact learning monotone {DNF} from membership queries.
\newblock In \emph{International Conference on Algorithmic Learning Theory}, pages 111--124. Springer, 2014.

\bibitem[Abasi et~al.(2018)Abasi, Bshouty, and Mazzawi]{Abasi2015}
Hasan Abasi, Nader~H. Bshouty, and Hanna Mazzawi.
\newblock Non-adaptive learning of a hidden hypergraph.
\newblock \emph{Theoretical Computer Science}, 716:\penalty0 15 -- 27, 2018.
\newblock ISSN 0304-3975.
\newblock \doi{https://doi.org/10.1016/j.tcs.2017.11.019}.
\newblock URL \url{http://www.sciencedirect.com/science/article/pii/S0304397517308496}.
\newblock Special Issue on ALT 2015.

\bibitem[Alon et~al.(2004)Alon, Beigel, Kasif, Rudich, and Sudakov]{alon2004learning}
Noga Alon, Richard Beigel, Simon Kasif, Steven Rudich, and Benny Sudakov.
\newblock Learning a hidden matching.
\newblock \emph{SIAM Journal on Computing}, 33\penalty0 (2):\penalty0 487--501, 2004.

\bibitem[Angluin(1988)]{angluin1988queries}
Dana Angluin.
\newblock Queries and concept learning.
\newblock \emph{Machine learning}, 2\penalty0 (4):\penalty0 319--342, 1988.

\bibitem[Angluin and Chen(2004)]{angluin2004learning}
Dana Angluin and Jiang Chen.
\newblock Learning a hidden graph using o (log n) queries per edge.
\newblock In \emph{International Conference on Computational Learning Theory}, pages 210--223. Springer, 2004.

\bibitem[Angluin and Chen(2006)]{Angluin2006}
Dana Angluin and Jiang Chen.
\newblock Learning a hidden hypergraph.
\newblock \emph{Journal of Machine Learning Research}, 7:\penalty0 2215–2236, December 2006.
\newblock ISSN 1532-4435.

\bibitem[Beigel et~al.(2001)Beigel, Alon, Kasif, Apaydin, and Fortnow]{beigel2001optimal}
Richard Beigel, Noga Alon, Simon Kasif, Mehmet~Serkan Apaydin, and Lance Fortnow.
\newblock An optimal procedure for gap closing in whole genome shotgun sequencing.
\newblock In \emph{Proceedings of the fifth annual international conference on Computational biology}, pages 22--30, 2001.

\bibitem[Bshouty(2018)]{bshouty2018exact}
Nader~H Bshouty.
\newblock Exact learning from an honest teacher that answers membership queries.
\newblock \emph{Theoretical Computer Science}, 733:\penalty0 4--43, 2018.

\bibitem[Chen et~al.(2008)Chen, Fu, and Hwang]{chen2008upper}
Hong-Bin Chen, Hung-Lin Fu, and Frank~K Hwang.
\newblock An upper bound of the number of tests in pooling designs for the error-tolerant complex model.
\newblock \emph{Optimization Letters}, 2\penalty0 (3):\penalty0 425--431, 2008.

\bibitem[Chin et~al.(2013)Chin, Leung, and Yiu]{chin2013non}
Francis~YL Chin, Henry~CM Leung, and Siu-Ming Yiu.
\newblock Non-adaptive complex group testing with multiple positive sets.
\newblock \emph{Theoretical Computer Science}, 505:\penalty0 11--18, 2013.

\bibitem[Chodoriwsky and Moura(2015)]{Chodoriwsky2015}
Jacob Chodoriwsky and Lucia Moura.
\newblock An adaptive algorithm for group testing for complexes.
\newblock \emph{Theoretical Computer Science}, 592:\penalty0 1 -- 8, 2015.
\newblock ISSN 0304-3975.
\newblock \doi{https://doi.org/10.1016/j.tcs.2015.05.005}.
\newblock URL \url{http://www.sciencedirect.com/science/article/pii/S0304397515003965}.

\bibitem[{D'yachkov} et~al.(2016){D'yachkov}, {Vorobyev}, {Polyanskii}, and {Shchukin}]{Dyachkov2016}
A.~G. {D'yachkov}, I.~V. {Vorobyev}, N.~A. {Polyanskii}, and V.~Y. {Shchukin}.
\newblock On multistage learning a hidden hypergraph.
\newblock In \emph{2016 IEEE International Symposium on Information Theory (ISIT)}, pages 1178--1182, 2016.

\bibitem[Flobak et~al.(2019)Flobak, Niederdorfer, Nakstad, Thommesen, Klinkenberg, and L{\ae}greid]{flobak2019high}
{\AA}smund Flobak, Barbara Niederdorfer, Vu~To Nakstad, Liv Thommesen, Geir Klinkenberg, and Astrid L{\ae}greid.
\newblock A high-throughput drug combination screen of targeted small molecule inhibitors in cancer cell lines.
\newblock \emph{Scientific data}, 6\penalty0 (1):\penalty0 1--10, 2019.

\bibitem[Gao et~al.(2006)Gao, Hwang, Thai, Wu, and Znati]{gao2006construction}
Hong Gao, Frank~K Hwang, My~T Thai, Weili Wu, and Taieb Znati.
\newblock Construction of d (h)-disjunct matrix for group testing in hypergraphs.
\newblock \emph{Journal of Combinatorial Optimization}, 12\penalty0 (3):\penalty0 297--301, 2006.

\bibitem[Hwang and Du(2006)]{hwang2006pooling}
Frank Kwang-ming Hwang and Ding-zhu Du.
\newblock \emph{Pooling designs and nonadaptive group testing: important tools for DNA sequencing}, volume~18.
\newblock World Scientific, 2006.

\bibitem[Menden et~al.(2019)Menden, Wang, Mason, Szalai, Bulusu, Guan, Yu, Kang, Jeon, Wolfinger, et~al.]{menden2019community}
Michael~P Menden, Dennis Wang, Mike~J Mason, Bence Szalai, Krishna~C Bulusu, Yuanfang Guan, Thomas Yu, Jaewoo Kang, Minji Jeon, Russ Wolfinger, et~al.
\newblock Community assessment to advance computational prediction of cancer drug combinations in a pharmacogenomic screen.
\newblock \emph{Nature communications}, 10\penalty0 (1):\penalty0 1--17, 2019.

\bibitem[Torney(1999)]{Torney1999}
David~C. Torney.
\newblock Sets pooling designs.
\newblock \emph{Annals of Combinatorics}, 3\penalty0 (1):\penalty0 95--101, 1999.
\newblock \doi{10.1007/BF01609879}.
\newblock URL \url{https://doi.org/10.1007/BF01609879}.

\end{thebibliography}

\newpage

\appendix

\section{Missing Analysis for Hypermatchings (Section~\ref{sec:hypermatchings})}
\label{appendix:hypermatching}

\subsection{Missing analysis for algorithm for learning hypermatchings (Section~\ref{subsection:learning_hypermatching})}
\label{appendix:learning_hypermatching}

\lemfuep*

\begin{proof}
It is clear that $\textsc{FindEdgeParallel}(S)$ makes at most $|S| + 1$ queries. If $S$ does not contain an edge, then $Q_M(S) = 0$ and the algorithm returns None. If $S$ contains at least two edges, then there is no intersection between the edges because $M$ is a matching. Therefore, for every $v \in S$, $Q_M(S \setminus \{v\}) = 1$, and the algorithm returns None. The only case left is when $S$ contains a unique edge $e$. In this case $Q_M(S) = 1$ and $e = \{ v \in S, \ Q_M(S \setminus \{v\} = 0\}$. In this case, $\textsc{FindEdgeParallel}(S)$ return $e$.
\end{proof}

Let $\S^i = \{(S^i_1, T^i_1), \ldots, (S^i_\ell, T^i_\ell)\}$ be $\S$ at the beginning of the $i$th iteration of the while loop of \textsc{FindEdgeAdaptive}$(S,s)$. To prove Lemma \ref{lem:fvcorrect}, we need the following three loop invariants in \textsc{FindEdgeAdaptive}. 
\lemfvconditions*
\begin{proof}
To simplify the proof, we actually show the following three loop invariants, where the second is an intermediary step to show the third invariant:

\begin{enumerate}
    \item either $v \in \hat{e}$ or $v \in S^i_j$ for some $j \in [\ell]$
    \item $Q_M(S^i_j \cup T^i_j) = 1$ for all $j \in [\ell]$
    \item $e \cap S^i_j \neq \emptyset$ for all $j \in [\ell]$
\end{enumerate}
We prove the three invariants by induction on $i$.

\vspace{2mm}

Base case: At the beginning of the first iteration, $\S^1 = \{(S,\emptyset)\}$ and $\hat{e} = \emptyset$, and we have the following 1) For every $v \in e$ we have $v \in S$. 2) $Q_M(S\cup \emptyset) = 1.$ 3) $e \cap S = e \neq \emptyset$.

 Assume that the statement of the lemma holds at the beginning of iteration $i \geq 1$. We will show that it holds at the beginning of iteration $i+1$. Let $\S^{i+1} = \{(S^{i+1}_1, T^{i+1}_1), \ldots, (S^{i+1}_k, T^{i+1}_k)\}$ be $\S$ at the beginning of the $(i+1)$th iteration of the while loop of \textsc{FindEdgeAdaptive}$(S,s)$.
 
 \begin{enumerate}
     \item Consider a vertex $v \in e$. If at the beginning of the $i$th iteration we had $v \in \hat{e}$, then $v \in \hat{e}$ at the beginning of the $(i+1)$th iteration as well. Suppose now that $v \in S^i_j$ for some $j \in [\ell]$ at the start of iteration $i$. If $|S^i_j| = 1$ then $v$ is added to $\hat{e}$ at Step 8. Assume $|S^i_j| > 1$, and let $T = T^i_j$. Then $S^i_j$ is partitioned into two sets $S_1$ and $S_2$. If $Q_M(S_1 \cup T) = 1$ and $Q_M(S_2 \cup T) = 1$ then this contradicts the fact that $S$ contains a unique edge. In fact, without loss of generality, we can assume that $v \in S_1$. Therefore $S_2\cup T$ cannot include $e$, and the only way $Q_M(S_2 \cup T) = 1$ is if $S$ contains another edge. If $Q_M(S_1 \cup T) = 1$ and $Q_M(S_2 \cup T) = 0$ then we know that $v \in S_1$ and $(S_1,T)$ is added to $\S^{i+1}$. If $Q_M(S_1 \cup T) = 0$ and $Q_M(S_2 \cup T) = 1$ then we know that $v \in S_2$ and $(S_2,T)$ is added to $\S^{i+1}$. If $Q_M(S_1 \cup T) = 0$ and $Q_M(S_2 \cup T) = 0$ then both  $(S_1, S_2 \cup T)$ and $(S_2, S_1 \cup T)$ are added to $\S^{i+1}$, and we have either $v \in S_1$ or $v \in S_2$.
     \item From steps 9 to 11 in \textsc{FindEdgeAdaptive}$(S,s)$, a couple $(S^{i+1}_j, T^{i+1}_j)$ is only added to $\S^{i+1}$ when $Q_M(S^{i+1}_j\cup T^{i+1}_j) = 1$. In step 12, we have $Q_M(S_1 \cup T) = 0$ and $Q_M(S_2 \cup T) = 0$, and both  $(S_1, S_2 \cup T)$ and $(S_2, S_1 \cup T)$ are added to $\S^{i+1}$. But we know from the induction assumption that at the beginning of iteration $i$ we had $(S_1 \cup S_2,T) \in \S^i$, therefore by 2. $Q_M(S_1 \cup S_2 \cup T) = 1$. This ensures that also in this case, $(S^{i+1}_j, T^{i+1}_j)$ is only added to $\S^{i+1}$ when $Q_M(S^{i+1}_j\cup T^{i+1}_j) = 1$.
     \item We know from the induction hypothesis $e \cap S^i_j \neq \emptyset$ for all $j \in [\ell]$. For a fixed $j \in [\ell],$ $S^i_j$ contains a vertex $v$ from $e$. In the while loop, $S^i_j$ is partitioned into $S_1$ and $S_2$, such that $v$ is either in $S_1$ or in $S_2$. If $Q_M(S_1 \cup T^i_j) = 1$, then $v \in S_1$, $(S_1,T^i_j)$ is added to $\S^{i+1}$ and $S_1 \cap e \neq \emptyset$. Similarly, if $Q_M(S_2 \cup T^i_j) = 1$, then $v \in S_2$, $(S_2,T^i_j)$ is added to $\S^{i+1}$ and $S_2 \cap e \neq \emptyset$. Assume now that $Q_M(S_1 \cup T^i_j) = Q_M(S_2 \cup T^i_j) = 0$, then $(S_1, S_2 \cup T^i_j)$ and $(S_2,S_1 \cup T^i_j$) are added to $\S^{i+1}$ and we need to show that $S_1 \cap e \neq \emptyset$ and $S_2 \cap e \neq \emptyset$ . We know from 2. that $Q_M(S_1\cup S_2\cup T^i_j) = 1$. Therefore $Q_M(S_1 \cup T^i_j) = Q_M(S_2 \cup T^i_j) = 0$ implies that there is a vertex of $e$ both in $S_1$ and $S_2$, therefore $S_1 \cap e \neq \emptyset$ and $S_2 \cap e \neq \emptyset$. This concludes this part and shows that whenever $(S_j^{i+1}, T_j^{i+1})$ is added to $\S^{i+1}$, we have $S_j^{i+1} \cap e \neq \emptyset$.
 \end{enumerate}
\end{proof}

\lemmafindedges*
\begin{proof}
We know that a sample $S_i$ contains at least an edge if and only if $Q_M(S_i) = 1$. We also know from Lemma \ref{lem:fuep} and Lemma \ref{lem:fvcorrect} that if $S_i$ contains a unique edge $e$ of size less than $s$, then the subroutine $\textsc{FindEdge}(S_i, s )$ will return $e$. Therefore, to learn all edges of size in $\left[s/\alpha, s \right]$, we only need to make sure that each such edge appears in at least one sample that does not contain any other edges, in other words, that $\S$ is a unique-edge covering family for $M_{s,\alpha,V'} = \{ e \in M: e \subseteq V', \ s/\alpha \leq |e| \leq  s$\}. Below, we show that it is the case with probability $1-O(n^{-\log n}) = 1 -e^{-\Omega(\log^2 n)}$.

We first show that with high probability, each edge $e$ of size $|e| \in \left[s/\alpha, s \right]$ is contained in at least $(\log^2 n)/2$ sample $S_i$'s. We use $X_e$ to denote the number of samples containing $e$, then we have
$$\E(X_e) = n^{\alpha} \log^2 n \cdot n^{-\frac{\alpha|e|}{s}} \geq \log^2 n .$$
By Chernoff bound, we have
$$P(X_e \leq \log^2 n/2) \leq P(X_e \leq \E(X_e)/2) \leq e^{-\frac{\E(X_e)}{8}} \leq n^{-\frac{\log n}{8}} .$$
As there are at most $\alpha n/s$ edges of size between $s/\alpha$ and $s$, by a union bound, we have

\[P(\exists e \in M \text{ s.t. } |e|\in \left[s/\alpha, s \right], X_e \leq (\log^2 n)/2) \leq \frac{\alpha n}{s}n^{-\log n/8} \leq n^{-(\log n/8-1)}.\]
Subsequently, 
\[P(\forall e \in M \text{ s.t. } |e|\in \left[s/\alpha, s \right], X_e \geq (\log^2 n)/2) \geq 1 - n^{-(\log n/8-1)}.\]











From now on we condition on the event that all edges $e$ whose size is between $s/\alpha$ and $s$ are included in at least $(\log^2 n) / 2$ samples. We show below that given $e \subseteq S_i$ for a sample $S_i$, the probability that there exists another edge of size at least $s$ in $S_i$ is upper bounded by $n^{-1/n}$. Recall that $M$ is the hidden matching we want to learn. We abuse notation and use $e' \in M \cap S_i$ to denote that an edge $e'$ is both in the matching $M$ and included in the sample $S_i$. We have
\begin{align}
\P(\exists e' \in M \cap S_i, e' \neq e \ | \ e \subseteq S_i) \leq & \sum_{e'\in M, e' \neq e} \P(e' \subseteq S_i \ | \ e \subseteq S_i)\label{eq:ad1}\\
= & \sum_{e'\in M, e' \neq e} \P(e' \subseteq S_i)\label{eq:ad2}\\
\leq & \frac{\alpha n}{s} \cdot (n^{-\frac{\alpha}{s}})^{\frac{s}{\alpha}}\label{eq:ad3}\\
= & \frac{\alpha n}{s} \cdot n^{-{\frac{\alpha s}{\alpha s}}} = \frac{\alpha}{s}. \notag
\end{align}
where Eq.\eqref{eq:ad1} uses union bound. Eq.\eqref{eq:ad2} is due to the fact that $M$ is a matching and thus $e\cap e' = \emptyset$ and vertices are sampled independently. Eq.\eqref{eq:ad3} follows because the total number of remaining edges is upper bounded by $\alpha n/s$ and each edge is in $S_i$ with probability at most $(n^{-\frac{\alpha}{s}})^{\frac{s}{\alpha}}$.\\

As each edge $e$ with size between $s/\alpha$ and $s$ is contained in at least $(\log^2 n)/2$ samples, we have that
\begin{align*}
\P(\forall S_i \text{ s.t. } e \subseteq S_i, \exists \ e' \in M \cap S_i, e' \neq e) \leq (\frac{s}{\alpha})^{-\log^2 n /2} \leq n^{-\log n/4},
\end{align*}

where the last inequality follows because $s/\alpha \geq 2$.

By another union bound on all edges of size between $s/\alpha$ and $ s $ (at most $n$ of them), we have that 
\begin{align*}
\P(\exists e \text{ s.t. } |e| \in \left[s/\alpha, s\right], \forall S_i \text{ s.t. } e \subseteq S_i, \exists e' \in M \cap S_i, e' \neq e) \leq n^{-(\log n/2-1)}. 
\end{align*}
We can thus conclude that with probability at least $1-O(n^{-\log n}) = 1 - e^{-\Omega(\log^2 n)}$, for all $e \in M$ with size between $s/\alpha$ and $s$, there exists at least one sample $S_i$ that contains $e$ but no other remaining edges. By Lemma \ref{lem:fuep} and Lemma \ref{lem:fvcorrect}, $e$ is returned by $\textsc{FindEdge}(S_i,s)$ and $e$ is added to the matching $\hat{M}$ that is returned by $\textsc{FindDisjointEdges}(s,\alpha,V')$.\\

Next, we show that $\textsc{FindDisjointEdges}$ is an $(r+1)$-adaptive algorithm with $n^\alpha \log^2 n + 2\alpha q n^\alpha \log^2 n/s$ queries. We first observe that $\textsc{FindDisjointEdges}$ makes only parallel calls to $\textsc{FindEdges}$ (after verifying that $Q_M(S_i) = 1$). Therefore, since $\textsc{FindEdges}$ runs in at most $r$ rounds, we get that $\textsc{FindDisjointEdges}$ runs in at most $r+1$ rounds. To prove a bound on the number of queries, we first argue using a Chernoff bound that with probability $1 - e^{-\Omega(\log^2 n)}$, every edge $e$ of size at least $s/\alpha$ is in at most $2 n^{\alpha-1} \log^2 n$ samples $S_i$. For a fixed edge $e$ we have:




\begin{align*}
    P(X_e > 2 n^{\alpha-1} \log^2 n) & \leq P(X_e > 2\cdot \E(X_e))\\
    & \leq  e^{\frac{-\E(X_e)}{3}} \\
    &  \leq  e^{\frac{-\log^2 n}{3}} = e^{-\Omega(\log^2 n)}.
\end{align*}

A union bound on the at most $n$ edges yields that every edge of size at least $s/\alpha$ is in at most $2 n^{\alpha-1} \log^2 n$ samples with probability $1-e^{-\Omega(\log^2 n)}$. Therefore, there are at most $\frac{2\alpha n}{s} n^{\alpha - 1} \log^2n$ samples $S_i$ such that $Q_M(S_i) = 1$ For every one of these samples, we call $\textsc{FindEdges}(S_i,s)$, which makes $q$ queries by assumption. Therefore, with high probability $1 - e^{-\Omega(\log^2 n)}$, $\textsc{FindDisjointEdges}$ makes $\O(n^{\alpha} \log^2 n + \alpha q  \frac{2n^{\alpha}}{s} \log^2n)$ queries.
\end{proof}

The following claim is a technical result we need to prove Theorem \ref{thm:learning_hypermatching_new}.

\begin{claim}\label{claim:1exponential_inequality}
For $0 \leq x < 1$ we have $\frac{1}{1-x} \geq e^x$.
\end{claim}
\begin{proof}
Consider the function $f(x) = \frac{1}{1-x} - e^x$ for $x < 1$. The derivative of $f$ is
\[ f'(x) = \frac{1}{(1-x)^2} - e^x = \frac{1}{1-x} \cdot \left(\frac{1}{1-x} - (1-x)e^x\right).\]

Because $1-x \leq e^{-x}$, we have 
$(1-x)e^x \leq 1$. Therefore,
\[ f'(x) \geq \frac{1}{1-x} \cdot \left(\frac{1}{1-x} - 1\right) \geq 0,\]
where the last inequality follows from the assumption that $0 \leq x < 1$.
\end{proof}

\thmlearninghypermatchingnew*

\begin{proof}
\textsc{FindSingletons} learns, with probability 1, all the edges of size $s = 1$. For edges of size greater than 2, from Lemma~\ref{lemma:find_disjoint_edjes}, we know that each call of $\textsc{FindDisjointEdges}(s,\alpha, V', \textsc{FindEdge})$ fails to find all edges in $V'$ of size between $s/\alpha$ and $s$ with probability at most $e^{-\Omega(\log^2 n)}$. As there are at most $n$ different edge sizes, $\textsc{FindDisjointEdges}$ is called at most $n$ times (we show below that it is actually called at most $\log n / (\log \alpha)$ times). Thus, the probability that at least one of the calls fails is upper bounded by 
\[n \cdot e^{-\Omega(\log^2 n)} = e^{\log n -\Omega(\log^2 n)}) = e^{-\Omega(\log^2 n)}.\]
We can thus conclude that the probability that all calls are successful is at least $1-e^{-\Omega(\log^2 n)}$.

Next, we show that $\textsc{FindMatching}$ makes $\O(\frac{\log n}{\log \alpha})$ calls to $\textsc{FindDisjointEdges}$.

We use $s_t$ to denote the value of $s$ after the $t^{\text{th}}$ call of $\textsc{FindDisjointEdges}$ for $t = 1, 2, \ldots, T$ with $T$ being the number of adaptive rounds of $\textsc{FindMatching}$ and set $s_0 = 2\alpha$. We show that the algorithm needs less than $O(\frac{\log n}{\log \alpha})$ rounds to go over all the possible edge sizes of the matching. In step 5 of $\textsc{FindMatching}$, we update $s_t$ as follows:
    \[ s_t \leftarrow  \lfloor \alpha s_{t-1} \rfloor + 1\]
    Therefore 
    \begin{align} \label{eq:ada-1}
        s_t & \geq \alpha s_{t-1},
    \end{align}
    and
    \begin{align} \label{eq:ada-2}
        s_t & \geq (\alpha)^t s_{0},
    \end{align}
    
Since the maximum $s$ we input to $\textsc{FindDisjointEdges}$ is $n$, we have that 
\[n \geq s_T \geq 2 (\alpha)^{T+1}\] and subsequently
\[T \leq \frac{\log n - \log 2}{\log \alpha}.\]

Therefore, $\textsc{FindMatching}$ makes $\O(\frac{\log n}{\log \alpha})$ calls to $\textsc{FindDisjointEdges}$. From Lemma \ref{lemma:find_disjoint_edjes}, we know that with high probability $1-e^{-\Omega(\log^2 n)}$, $\textsc{FindDisjointEdges}(s,\alpha,V', \textsc{FindEdge})$ runs in $r+1$ rounds and  makes $\O(n^{\alpha} \log^2 n + 2\alpha q n^{\alpha} \log^2 n / s)$ queries when \textsc{FindEdges}$(S,s)$ runs in $r$ rounds and makes $q$ queries. From Lemma \ref{lem:fuep} and Lemma \ref{lem:fvcorrect} we know that for the subroutines \textsc{FindEdgeAdaptive} and \textsc{FindEdgeParallel} we have $r = \O(\log n), q = \O(s\log n)$ and $r = 1,q = \O(n)$ respectively. Therefore $\textsc{FindDisjointEdges}$ runs in $\O(\log n)$ rounds and makes $\O(n^{\alpha} \log^2 n + 2n^{\alpha}\log^3n)$ queries with \textsc{FindEdgeAdaptive} and in 2 rounds $\O(n^{\alpha} \log^2 n + 2n^{\alpha + 1}\log^2n)$ with \textsc{FindEdgeParallel}.
We conclude that with high probability, 
\begin{itemize}
    \item 
 $\textsc{FindMatching}$ runs in $\O(\frac{\log^2 n}{\log \alpha})$ rounds and  makes $\O(n^{\alpha} \frac{\log^3 n}{\log \alpha} + 2\alpha n^{\alpha} \frac{\log^4n}{\log \alpha})$ queries with \textsc{FindEdgeAdaptive}, \item$\textsc{FindMatching}$ runs in $\O(\frac{\log n}{\log \alpha})$ rounds and makes  $\O(n^{\alpha} \frac{\log^3 n}{\log \alpha} + 2n^{\alpha+1}\frac{\log^3n}{\log \alpha})$
queries with \textsc{FindEdgeParallel}.
\end{itemize}

Now we turn our attention to particular values of $\alpha$.
\begin{itemize}
    \item Suppose now that $\alpha = 1/(1 - \frac{1}{2\log n})$ and that we use \textsc{FindEdgeAdaptive}. We get then $\alpha = 1 + o(1)$. Furthermore, by Claim \ref{claim:1exponential_inequality}, we have 
\begin{align*}
     \log \alpha & = \log \frac{1}{1 - \frac{1}{2\log n}}\\
     & \geq \frac{1}{2\log n}.
\end{align*}
Therefore we get that $\frac{1}{\log \alpha} = O(\log n)$. We also observe that for $n \geq 3$ 
    \begin{align*}
     \alpha & = 1 + \frac{\frac{1}{2\log n}}{1 -\frac{1}{2\log n}}\\
     & \leq 1 + \frac{1}{\log n}.
     \end{align*}

Therefore $n^{\alpha} \leq n \cdot n^{\frac{1}{\log n}} = \O(n)$. We finally conclude that when $\alpha = 1/(1 - \frac{1}{2\log n})$ and we use the subroutine \textsc{FindEdgeAdaptive}, $\textsc{FindMatching}$ runs in $\O(\log^2 n/\log \alpha) = \O(\log^3 n)$ adaptive rounds and makes $\O(n\log^4 n + n \log^5 n) = \O(n \log^5 n)$ queries in total.
\item With $\alpha = 2$ and \textsc{FindEdgeParallel}, we have $r = 1$ and we get that $\textsc{FindMatching}$ learns any hypermatching w.h.p. in $O(\log n)$ rounds and with at most $\O(n^{3} \log^3 n)$ queries.
\end{itemize}
\end{proof}

\subsection{Missing analysis for hardness of learning hypermatchings (Section~\ref{subsection:hardness_hypermatching})}
\label{appendix:hardness_hypermatching}

\subsubsection{Lower bound for non-adaptive algorithms}
To argue that learning hypermatchings non-adaptively requires an exponential number of queries, we fix  the number of vertices $n$, we construct a family of matchings $M_P$ which depend on a partition $P$ of the vertex set $[n]$ into 3 parts $P_1,P_2,P_3$ such that 
\begin{itemize}
        \item $|P_1| = n - (\sqrt{n} + 1)$.
         \item $P_1$ contains a perfect matching $M_1$ with edges of size 2.
        \item $|P_2| = 1$
        \item $P_3$ is an edge of the matching s.t. $|P_3| = \sqrt{n}$.
        \item $M_P = M_1 \cup \{P_3\}$.
        \end{itemize}

We denote all such partitions by $\mathcal{P}_3$. For a partition $P \in \mathcal{P}_3$, multiple matchings are possible, depending on the perfect matching $M_1$. We use $M_P$ to denote a random matching from all the possible matchings satisfying the properties above. The main idea of the construction is that after one round of queries, elements in $P_2$ and $P_3$ are indistinguishable to any algorithm with polynomial queries. However, a learning algorithm needs to distinguish elements in $P_3$ from the element in $P_2$ to learn the edge $P_3$.

Before presenting the proof, we formalize what it means for a set of elements to be indistinguishable. Since the next  definition is also used in the next subsection, we consider an arbitrary family $\mathcal{P}_R$ of partitions $P = (P_0, \ldots, P_R)$. Given a partition $(P_0, \ldots, P_R)$, we denote by $P_r$: the union of parts $P_i$ such that $i \geq r$, i.e. $P_r := \cup_{i=r}^R P_i$. Informally, we say that queries $Q_{M_P}(S)$ are independent of the partition of $P_r$: if the values $Q_{M_P}(S)$ of these queries do not contain information about which elements are in $P_r$, or $P_{r+1}, \ldots ,$ or $P_R$.

\begin{definition}
Given a family of partitions $\mathcal{P}_R$ and a partition $P = (P_0,\ldots,P_R) \in \mathcal{P}_R$, let $P'$ be
a partition chosen uniformly at random from $\{ (P_0',\ldots,P_R') \in \mathcal{P}_3 \ : \ P_0' = P_0, \ldots P'_{r-1} = P_{r-1} \}$. Let $M_{P'}$ be a matching on $P'$ such that $M_P$ and $M_{P'}$ are equal on $P_r := \cup_{i=r}^R P_i$. A query $Q_{M_P}(S)$ is independent from $P_r := \cup_{i=r}^R P_i$ if $Q_{M_P}(S) = Q_{M_{P'}}(S)$ with high probability over $P'$.
\end{definition}

For example, in the partition $(P_1,P_2,P_3)$ above, any query $Q_{M_P}(S)$ such that $S$ contains an edge from $M_1$ is independent of the partition of $P_2 \cup P_3$ since this query will always answer 1.

\textbf{Analysis of the construction.} We consider a non-adaptive algorithm $\mathcal{A}$, a uniformly random partition $P = (P_1, P_2, P_3)$ in $\mathcal{P}_3$, and a matching $M_P$. We argue that with high probability over both the randomization of $P$ and the decisions of $A$, the matching returned by $\mathcal{A}$ is not equal to $M_P$. The analysis consists of two main parts.

The first part argues that, with high probability, for all queries $S$ made by a non-adaptive algorithm, $Q_{M_P}(S)$ is independent of the partition of $P_2 \cup P_3$.

\begin{lemma}\label{lemma:non_adaptive_hardness_1}
Let $P$ be a uniformly random partition from $\mathcal{P}_3$. For any collection $\mathcal{S}$ of $\poly(n)$ non-adaptive queries, with high probability $1-e^{-\Omega(n)}$, for all $S \in \mathcal{S}$, $S$ is independent of the partition
of $P_2 \cup P_3$.\end{lemma}

\begin{proof}
Consider any set $S \in \mathcal{S}$ which is independent of the randomization of $P$. There are two cases
depending on the size of $S$.
\begin{enumerate}
    \item $|S| \leq \frac{n + \sqrt{n} + 1 }{2}$. In this case, if $S$ does not contain at least $\sqrt{n}$ vertices from $P_2 \cup P_3$, then for any other partition $P'$ such that $P_1' = P_1$, we have $Q_{M_P}(S) = Q_{M_{P'}}(S)$. In fact, both $Q_{M_{P'}}(S)$ and $Q_{M_P}(S)$ will be equal 1 if and only if $S$ contains an edge from $M_1$. Next we show that, with high probability,  $S$ will contain fewer than $\sqrt{n}$ vertices from $P_2 \cup P_3$. Since $P_2 \cup P_3$ is a uniformly random set of size $\sqrt{n}+1$, by the Chernoff bound, $|S \cap (P_2 \cup P_3)| < \sqrt{n}+1$ with probability $1-e^{-\Omega(\sqrt{n})}$.

    We therefore get that $Q_{M_P}(S)=Q_{M_{P'}}(S)$ for any partition $P' =(P_1,P_2',P_3')$ and query $|S| \leq \frac{n + \sqrt{n}+1}{2}$ with probability $1-e^{-\Omega(\sqrt{n}))}$.
    \item $|S| > \frac{n + \sqrt{n} + 1 }{2}$. In this case, the set $S$ contains at least one matched edge from $P_1$. In fact, the number of edges in $M_1$ plus the size of $P_2 \cup P_3$ is
    \[ \frac{|P_1|}{2} + |P_2| + |P_3| = \frac{n -(\sqrt{n} + 1)}{2} + \sqrt{n} + 1 = \frac{n + \sqrt{n} + 1 }{2} < |S|.\]
    
    $S$ must therefore contain at least one edge from $M_1$, and we get that $Q_{M_P}(S)=Q^{P'}(S)$ for any partition $P' =(P_1,P_2',P_3')$ and query $|S| > \frac{n + \sqrt{n}+1}{2}$ with probability 1.
\end{enumerate}

By combining the two cases $|S| \leq \frac{n + \sqrt{n} + 1 }{2}$ and $|S| > \frac{n + \sqrt{n} + 1 }{2}$ , we get that with probability $1-e^{-\Omega(\sqrt{n}))}$, $Q_{M_P}(S)$ is independent of the partition of $P_2 \cup P_3$, and by a union bound, this holds for any collection of $poly(n)$ sets $S$.
\end{proof}

The second part of the analysis argues that if all queries $Q_{M_P}(S)$ of an algorithm are independent of the partition of $P_2 \cup P_3$, then, with high probability, the matching returned by this algorithm is not equal to $M_P$.

\begin{lemma}\label{lemma:non_adaptive_hardness_2}
Let $P$ be a uniformly random partition in $\mathcal{P}_3$ and $M_P$ a random matching over $P$. Consider an algorithm $\mathcal{A}$ such that all the queries $Q_{M_P}(S)$ made by $\mathcal{A}$ are independent of the partition of $P_2 \cup P_3$. Then, the (possibly randomized) matching $M$ returned by $\mathcal{A}$ is, with probability $1 -O(1/\sqrt{n})$, not equal to $M_P$.
\end{lemma}

\begin{proof}
Consider an algorithm $\mathcal{A}$ such that all queries $Q_{M_P}(S)$ of $\mathcal{A}$ are independent of the partition of $P_2 \cup P_3$. Thus, the matching $M$ returned by $\mathcal{A}$ is conditionally independent of the randomization of the partition $P$ given $P_1$.

For the algorithm $\mathcal{A}$ to return $M_P$, it needs to learn the edge $P_3$. We distinguish two cases.
\begin{itemize}
    \item $\mathcal{A}$ does not return any edge that is included in $P_2 \cup P_3$. In this case, with probability 1, $M$ is not equal to $M_P$.
    \item $\mathcal{A}$ returns an edge that is included in $P_2 \cup P_3$. We know from the previous lemma that with probability $1 -e^{-\Omega(\sqrt{n})}$, all queries are independent from $P_2 \cup P_3$. Therefore, the edge returned by $\mathcal{A}$ is also independent of the partition $P_2 \cup P_3$. The probability that this edge is equal to $P_3$ is less than $1/\sqrt{n+1}$. Therefore, with probability $(1 -e^{-\Omega(\sqrt{n})})(1-1/\sqrt{n+1})= 1-O(1/\sqrt{n})$, the returned matching $M$ is not equal to $M_P$.
\end{itemize}

\end{proof}

By combining Lemma \ref{lemma:non_adaptive_hardness_1} and Lemma \ref{lemma:non_adaptive_hardness_2}, we get the hardness result for non-adaptive
algorithms.

\thmhardnessone*

\begin{proof}
Consider a uniformly random partition $P \in \mathcal{P}_3$, a matching $M_P$ and an algorithm $\mathcal{A}$ which queries $M_P$ . By Lemma \ref{lemma:non_adaptive_hardness_1}, after one round of queries, with probability $1 - e^{-\Omega(\sqrt{n})}$ over both the randomization of $P$ and of the algorithm, all the queries $Q_{M_P}(S)$ made by $\mathcal{A}$ are independent of the partition of $P_2 \cup P_3$. By Lemma \ref{lemma:non_adaptive_hardness_2}, this implies that, with probability $1 - O(1/\sqrt{n})$, the matching  returned by $\mathcal{A}$ is not a equal to $M_P$. By the probabilistic method, this implies that there exists a partition $P \in \mathcal{P}_3$ for which, with probability $1 - O(1/\sqrt{n})$, $\mathcal{A}$ does not return $M_P$ after one round of queries.
\end{proof}

\subsubsection{Lower bound for $o(\log \log n)$ adaptive algorithm}

Our idea is to construct a partition $P_0, \ldots, P_i, \ldots P_{R}$, such that every set $P_i$ is a perfect matching of edges of size $d_i$, and therefore has $|P_i|/d_i$ edges. We want to choose the size $|P_i|$ and $d_i$ ($d_i$ increasing) such that with probability $1 - e^{-poly(n)}$, after $i$ rounds of any adaptive algorithm that asks a polynomial number of queries, the partition $P_i \cup P_{i+1} \ldots P_R$ is indistinguishable.\\


\noindent \textbf{Our construction}
\begin{itemize}
    \item $d_0 = 3$.
    \item $d_{i+1} = 3\log^{2} n \cdot d_i^2$.
    \item $|P_{i}| = 3\log^{2} n \cdot d_i^2$.
    \item $|P_{i+1}| = 3\log^{2} n  \cdot d_{i+1}^2 = 3\log^{2}n \cdot (3\log^{2} n \cdot d_i^2)^2 = 3\log^{2} n  \cdot |P_i|^2 $.
\end{itemize}

\noindent 
\textbf{Notation} : Let $\mathcal{P}_R$ be the set of partitions $P = (P_0,\ldots, P_R)$ satisfying the conditions above. For a fixed $i$, let $k_i = \frac{|P_i|}{d_i}$, and let $e_1^i, \ldots, e_j^i, \ldots, e^i_{k_i}$ be the random matching on $P_i$. Finally, let $n_i = \sum\limits_{l \geq i} |P_l|$\\

\begin{claim}\label{claim:size_of_partition}
The construction has $R + 1 = \Theta( \log \log n)$ partitions.
\end{claim} 

\begin{proof}
We first show that the construction has $R +1 \geq \log \log n$ partitions.

By induction we have that \[ |P_i| = (3\log^{2} n)^{2^i - 1} |P_0|^{2^i}\]

Therefore 
\begin{align*}
    n & \leq \sum\limits_{i=0}^R |P_i|\\
    & \leq (R+1) |P_R|\\
    & \leq (R+1) (3\log^{2} n)^{2^R - 1} |3\log^n d_0^2|^{2^R}\\
    & = (R+1) 9^{2^R} (3\log^{2} n)^{2^{R+1} - 1}\\
    & = 9^{2^{R+1}} (3\log^{2} n)^{2^{R+1}}\\
\end{align*}

This implies that 
\[ 2^{R+1} \log(27\log^2 n) \geq \log n,\]
and 
\[ R+1 \geq  \frac{\log \log n}{\log 2} - \frac{\log \left(\log(27\log^2 n)\right)}{\log 2} \geq \log \log n.\]

Next we show that $R = O( \log \log n)$.  We know that  $|P_R| \leq n$, which implies that $(3\log^{2}n)^{2^R - 1} \leq n$ and \[
 2^R - 1 \leq \frac{\log n}{ \log \log^6 n}.\]
 Therefore, 
 \[ R = O(\log \frac{\log n}{ \log \log^6 n}) = O(\log \log n).\]
\end{proof}

Suppose that at the beginning of round $i$, we learned $P_0 \cup \ldots \cup P_{i-1}$ but $P_i \cup P_{i+1} \ldots P_R$ is indistinguishable. Consider a query $S$ of the $i$-th round. We show that if the query $S$ is big, then with high probability $S$ contains an edge from $P_i$, and if $S$ is small, then $S$ does not contain any edge from the high partition $P_{i+1}, P_{i+2}, \ldots P_R.$ In both cases, querying $S$ gives an answer that is independent to $P_{i+1}, P_{i+2}, \ldots P_R$ with high probability over all the polynomial number of queries. Therefore, after the $i$-th round, $P_{i+1}, P_{i+2}, \ldots $ is indistinguishable.\\

\textbf{Outline of the proof.} Claim \ref{claim:binomial_bounds} is a technical result we need to bound the binomial coefficients in our proofs. Lemma \ref{claim:intersection_expectation} gives the expected size of the intersection of a query with a partition and presents a concentration result on the intersection when the queries are large. Lemma \ref{claim:uniform_edges} shows that for the purpose of bounding the probability that no edge of a partition is included in $S$, we can drop the constraint the edges are disjoint and think of them as being randomly and uniformly sampled (with replacement) from the partition. Lemma \ref{lemma:large_query} shows that $|S| \geq (1-\frac{1}{d_i})n_i$, then with probability $1 - e^{-\Omega(\log^2 n)}$, $S$ contains at least one edge from the matching on $P_i$. If $|S| \leq (1-\frac{1}{d_i})n_i$, then Lemma \ref{lemma:small_query} shows with probability $1 - e^{-\Omega(\log^2 n)}$, for every $l \geq i+1$, $S$ does not contain any edge from the matching on $P_l$. Combining the two lemmas (Lemma \ref{lemma:induction_loglogn_hardness}) shows that if $P$ is a uniformly random partition in $P_R$, and if all the queries made by an algorithm $\A$ in the previous $i$ rounds are independent of the partition of $P_i, \ldots, P_R$, then any collection of $poly(n)$ non-adaptive queries at round $i+1$, independent of $P_{i+1}, \ldots, P_R$
with probability $1 - e^{-\Omega(\log^2 n)}$.\\

\begin{claim}\label{claim:binomial_bounds}
If $k < \sqrt{n}$, then
\[\frac{n^k}{4(k!)} \le {n \choose k}\le \frac{n^k}{k!}\]
\end{claim}

\begin{proof}
We have that 
\begin{align*}
    {n \choose k} & = \frac{n(n-1)\ldots (n-(k-1)}{k!}\\
    & = \frac{n^k}{k!}(1-\frac{1}{n})\ldots (1-\frac{k-1}{n})\\
\end{align*}
The upper bound ${n \choose k} \leq \frac{n^k}{k!}$ follows immediately. To show the lower bound, we observe that
\begin{align*}
    {n \choose k} & = \frac{n^k}{k!}(1-\frac{1}{n})\ldots (1-\frac{k-1}{n})\\
    & \geq \frac{n^k}{k!}(1-\frac{k-1}{n})^{k-1}
\end{align*}
Now consider the function $f_n(x,y) = (1-\frac{x}{n})^y$ for $y \ge 1$ and $x \in [1,n]$. For $x \in [1,n]$ and fixed $y\ge 1$, the function $f_n$ decreases as $x$ increases. Similarly, for $x$ fixed, $f_n$ decreases as $y$ increases. Therefore
\[(1-\frac{k-1}{n})^{k-1} \geq (1-\frac{1}{\sqrt{n}})^{\sqrt{n}}.\]
Furthermore, we know that for $x\geq 2$, $(1-1/x)^x \geq 1/4$. Therefore \begin{align*}
    {n \choose k} & \geq \frac{n^k}{k!}(1-\frac{k-1}{n})^{k-1}\\
    & \geq \frac{n^k}{k!}(1-\frac{1}{\sqrt{n}})^{\sqrt{n}}\\
    & \geq \frac{n^k}{4(k!)}.
\end{align*}
\end{proof}

\begin{lemma}\label{claim:intersection_expectation}
Let $S$ be an arbitrary subset of $P_i \cup P_{i+1} \cup \ldots P_R$. After $i$ rounds, for $l \geq i$, the expected size of the intersection $|S\cap P_l|$ is 
\[ \E[|S\cap P_l|] = s_l = \frac{|S||P_l|}{n_i},\]
where the expectation is over the partitions of $P_i \cup P_{i+1} \cup \ldots P_r$ into $P_i, P_{i+1}, \ldots, P_r$. Moreover, for any polynomial set of queries $\S$ such that $|S| \geq (1-\frac{1}{d_i})n_i$ for $S \in \S$, then with probability $1 - e^{-\Omega(\log^2 n)}$, for all $l \geq i$  and for all $S \in \S$, we have $|S\cap P_l| \geq s_l (1-\frac{1}{d_i})$.
\end{lemma}
\begin{proof}[Proof of Lemma \ref{claim:intersection_expectation}]
Every node in $S$ will be in $P_l$ with probability $|P_l|/n_i$, therefore  by linearity of expectation
\[ \E[|S\cap P_l|] = s_l = \frac{|S||P_l|}{n_i}\]
Suppose now that $|S| \geq (1-\frac{1}{d_i})n_i$. By Chernoff bound, we have
\begin{align*}
 \P\big(|S\cap P_l|  \leq s_l (1-\frac{1}{d_i})\big)
 & \leq e^{-\frac{s_l}{2d_i^2}}\\
 & = e^{-\frac{|S|}{n_i}\frac{|P_l|}{2d_i^2}}\\
 & \leq e^{-(1-1/d_i) \frac{|P_l|}{2d_i^2}}\\
 & \leq e^{- \frac{|P_l|}{3d_i^2}}\\
 & \leq e^{-\log^2 n}\\
 & \leq \frac{1}{n^{\log n}}
 \end{align*}
 
In the fourth inequality, we used $1- 1/d_i \leq 2/3$. By a union bound the probability that there exists a partition $P_l$ with $l \geq i$ such that $|S\cap P_l|  \leq s_l (1-\frac{1}{d_i})$ is $O(\log \log n / n^{\log n}) = e^{-\Omega(\log^2 n)}$. Another application of the union bound on the polynomially many queries shows that with probability at least $1 - e^{-\Omega(\log^2 n)}$, for all queries of size greater than $(1-1/d_i)n$ and for all partitions greater than $i$, we have $|S\cap P_l|  \geq s_l (1-\frac{1}{d_i})$.
\end{proof}

\begin{lemma}\label{claim:uniform_edges}
Fix an iteration $i$, and a query $S$. Let $e'_1, \ldots, e'_{k_i}$ be edges of size $d_i$, independently and uniformly  sampled from $P_i$. Let $e^i_1, \ldots, e^i_{k_i}$ be a random matching on $P_i$. We have 
\[ \P( \forall j \in \{1,\ldots, k_i\}, \ e_j^i \not\subseteq S) \leq \P( \forall j \in \{1,\ldots, k_i\}, \ e'_j \not\subseteq S) = (1 - \frac{\binom{|S\cap P_i|}{d_i}}{\binom{|P_i|}{d_i}})^{k_i} \]
\end{lemma}  

\begin{proof}
Since the partition $i$ is fixed, we drop indexing by $i$ and use $e_j$ to denote $e_j^i$, $k$ to denote $k_i$, and $d$ to denote $d_i$. We also let $n = |P_i|$ and assume without loss of generality that $S \subseteq |P_i|$.\\

The intuition is that, since the edges $e_j$ must be disjoint but the edges $e'_j$ do not necessarily have to, then the edges $e_j$ will cover a ``bigger fraction'' of $S$ than $e'_j$.\\

We prove this by induction on the number of edges.  When considering only one edge $e_1$, we know that  \[
\P( e_1 \not\subseteq S) = \P( e_j \not\subseteq S) = \P(e'_1 \not\subseteq S).\]

Now assume that the result is true for $j$ edges, i.e. $\P(e_1 \not\subseteq S, \ldots, e_j \not\subseteq S) \leq \P(e'_1 \not\subseteq S)^j$. We want to show that it holds for $j+1$ edges, i.e.,
\[\P(e_1 \not\subseteq S, \ldots, e_j \not\subseteq S, e_{j+1} \not\subseteq S) \leq \P(e'_1 \not\subseteq S)^{j+1} \]
By the inductive hypothesis, it is sufficient to show that 
\[ P(e_{j+1} \not\subseteq S | \ \forall \ l \leq j, \  e_l \not\subseteq S) \leq \P(e_{j+1} \not\subseteq S)\]
This is equivalent to showing that $P(e_{j+1} \subseteq S | \ e_1 \not\subseteq S, \ldots, e_j \not\subseteq S) \geq \P(e_{j+1} \subseteq S)$. Furthermore, we have 
\begin{align*}
     \P(e_{j+1} \subseteq S) & = \P(e_{j+1} \subseteq S | \ \forall \ l \leq j, \  e_l \not\subseteq S)\P( \forall \ l \leq j, \  e_l \not\subseteq S) \\
     & \ \ \ + \P(e_{j+1} \subseteq S | \ \exists \  l \leq j, \  e_l \subseteq S)\P( \exists \ l \leq j, \  e_l \subseteq S)
     \end{align*}
Therefore, $P(e_{j+1} \not\subseteq S | \ e_1 \not\subseteq S, \ldots, e_j \not\subseteq S) \leq \P(e_{j+1} \not\subseteq S)$ becomes equivalent to 
\begin{equation}\label{eq:proba_inequality}
 P(e_{j+1} \subseteq S | \ \forall \ l \leq j, \  e_l \not\subseteq S) \P( \exists \ l \leq j, \  e_l \subseteq S) \geq \P(e_{j+1} \subseteq S | \ \exists \  l \leq j, \  e_l \subseteq S)\P( \exists \ l \leq j, \  e_l \subseteq S)
 \end{equation}

Since, $\P( \exists \ l \leq j, \  e_l \subseteq S) > 0$, the inequality \eqref{eq:proba_inequality} is equivalent to 
\begin{equation}\label{eq:proba_inequality2}
 P(e_{j+1} \subseteq S | \ \forall \ l \leq j, \  e_l \not\subseteq S) \geq \P(e_{j+1} \subseteq S | \ \exists \  l \leq j, \  e_l \subseteq S)
 \end{equation}
 
To prove equation \eqref{eq:proba_inequality2}, we observe that 
\[ \P(e_{j+1} \subseteq S) = p P(e_{j+1} \subseteq S | \ \forall \ l \leq j, \  e_l \not\subseteq S) + (1-p) \P(e_{j+1} \subseteq S | \ \exists \  l \leq j, \  e_l \subseteq S),\]
where $p = \P(\forall \ l \leq j, \  e_l \not\subseteq S).$ Therefore, if we show that $\P(e_{j+1} \subseteq S) \geq \P(e_{j+1} \subseteq S | \ \exists \  l \leq j, \  e_l \subseteq S)$, then we must have $P(e_{j+1} \subseteq S | \ \forall \ l \leq j, \  e_l \not\subseteq S) \geq \P(e_{j+1} \subseteq S) \geq  \P(e_{j+1} \subseteq S | \ \exists \  l \leq j, \  e_l \subseteq S)$. To see that $\P(e_{j+1} \subseteq S) \geq \P(e_{j+1} \subseteq S | \ \exists \  l \leq j, \  e_l \subseteq S)$, we first observe that 
\begin{align*}
\P(e_{j+1} \subseteq S | \ \exists \  l \leq j, \  e_l \subseteq S) & = \P(e_{j+1} \subseteq S | \  e_1 \subseteq S)\\
& = \frac{\binom{|S| - d}{d}}{\binom{n -d}{d}},
\end{align*} 
while 
\[ \P(e_{j+1} \subseteq S) = \frac{\binom{|S|}{d}}{\binom{n}{d}}.\]
Therefore,
\begin{align*}
 \frac{\P(e_{j+1} \subseteq S)}{\P(e_{j+1} \subseteq S | \  e_1 \subseteq S)}  & =  \frac{\binom{|S|}{d}}{\binom{n}{d}} \frac{\binom{n -d}{d}}{\binom{|S| - d}{d}}\\
 & = \frac{(|S|)!}{(|S|-d)!} \frac{(|S|-2d)!}{(|S|-d)!} \frac{(n-d)!}{n!} \frac{(n-d)!}{(n-2d)!}\\
 & = \frac{\frac{|S| \ \ \ \ldots (|S|-d+1)}{(|S|-d) \ldots |S|-2d+1)}}{\frac{n \ \ \ \ldots (n-d+1)}{(n-d) \ldots n-2d+1)}}
\end{align*}
It is easy to verify that \[\frac{|S|-d + i}{|S| -2d +i} \geq \frac{n -d +i}{n -2d+ i}\]
for every $i \in \{1, \ldots, d\}$, because $|S| \leq n$. Therefore we get that $\P(e_{j+1} \subseteq S) \geq \P(e_{j+1} \subseteq S | \  e_1 \subseteq S)$, and hence 
$\P(e_{j+1} \subseteq S) \geq \P(e_{j+1} \subseteq S | \ \exists \  l \leq j, \  e_l \subseteq S)$. This proves \eqref{eq:proba_inequality2} and concludes the induction proof.
\end{proof}

\begin{lemma}\label{lemma:large_query}
Consider a query $S$ by an algorithm at round $i$ such that all the queries from previous rounds are independent of $P_{i},\ldots, P_R$. Let $S$ be such that $|S| \geq (1-\frac{1}{d_i})n_i$, then with probability $1 - e^{-\Omega(\log^2 n)}$, $S$ contains at least one edge from the matching on $P_i$
\end{lemma}
\begin{proof}
\begin{align}
    \P( \exists j \in \{1,\ldots, k_i\}, \ e_j^i \subseteq S) & = 1 -  \P( \forall j \in \{1,\ldots, k_i\}, \ e_j^i \not\subseteq S)\\
    & \geq 1 - (1 - \frac{\binom{|S\cap P_i|}{d_i}}{\binom{|P_i|}{d_i}})^{k_i}\label{eq:proba_existence}
\end{align}
where the inequality comes from Lemma \ref{claim:uniform_edges}. Since $|S| \geq (1-\frac{1}{d_i})n_i$, we have with high probability that $|S\cap P_i| \geq (1-\frac{1}{d_i})^2 |P_i| \geq d_i^2$. Conditioning on this event, and since we already have $d_i \leq \sqrt{|P_i|}$, we get by Claim \ref{claim:binomial_bounds}
\begin{align*}
 \frac{\binom{|S\cap P_i|}{d_i}}{\binom{|P_i|}{d_i}} &\geq \frac{|S\cap P_i|^{d_i}}{4(d_i)!} \frac{(d_i)!}{|P_i|^{d_i}}\\
 & \geq \frac{1}{4} (1-\frac{1}{d_i})^{2d_i}\\
 & \geq \frac{1}{4} e^{-\frac{1/d_i}{1-1/d_i} 2d_i}\\
 & \geq \frac{1}{4} e^{-\frac{2}{1-1/d_i} }\\
 & \geq \frac{1}{4} e^{-4}
 \end{align*}
 
 Since $k_i = \frac{|P_i|}{d_i} = 3\log^{2} n \cdot d_i$, we get that 
 \[\frac{\binom{|S\cap P_i|}{d_i}}{\binom{|P_i|}{d_i}} \geq \frac{1}{4} e^{-4} \geq \frac{1}{3 d_i} = \frac{\log^{2} n}{k_i}\]
 
 The probability \eqref{eq:proba_existence} becomes
 \begin{align}
    \P( \exists j \in \{1,\ldots, k_i\}, \ e_j^i \subseteq S) & \geq 1 - (1 - \frac{\binom{|S\cap P_i|}{d_i}}{\binom{|P_i|}{d_i}})^{k_i} \\
    & \geq 1 - (1 - \frac{\log^{2} n}{k_i})^{k_i}\\
    & \geq 1 - \frac{1}{n^{\log n}}\\
    & = 1 - e^{-\Omega(\log^2 n)}
\end{align}
\end{proof}

\begin{lemma}\label{lemma:small_query}
Consider a query $S$ by an algorithm at round $i$ such that all the queries from previous rounds are independent of $P_{i},\ldots, P_R$. Let $S$ be  such that $|S| \leq (1-\frac{1}{d_i})n_i$, then with probability $1 - e^{-\Omega(\log^2 n)}$, for every $l \geq i+1$, $S$ does not contain any edge from the matching on $P_l$.
\end{lemma}

\begin{proof}
Let's fix a partition $P_l$ with $l \geq i+1$. Let $e_1, \ldots, e_j, \ldots e_{k_l}$ be the random matching on $P_l$. 
We have for $j = 1, \ldots, k_l$,
\[ \E\big[ |e_j \cap S|\big] = \frac{|S||e_j|}{n_i} \leq (1-\frac{1}{d_i}) d_{l} = d_l - \frac{d_l}{d_i}.\]

Going forward, we present an upper bound on $\P(e_j \subset S)$. We can assume without loss of generality that $|S| = (1-\frac{1}{d_i})n_i$. In fact, for any $j = 1, \ldots, k_l$, if $S \subseteq S'$ then we have $\P(e_j \subseteq S) \leq \P(e_j \subseteq S')$. \\

The probability that $e_j \subset S$ can be expressed as 
\begin{align*}
    \P(e_j \subset S) & = \P(|e_j \cap S| \geq d_l)\\
    & \leq \P\Big(|e_j \cap S| \geq (1+\delta)E\big[ |e_j \cap S|\big] \Big),
\end{align*}
where $\delta = \frac{1}{d_i - 1}$.
By Chernoff bound we get 
\begin{align*}
    \P(e_j \subset S) &  \leq 2 e^{-\frac{1}{3} \delta^2 E\big[ |e_j \cap S|\big]}\\
    &  = 2 e^{-\frac{1}{3} \frac{1}{(d_i -1)^2} (1-\frac{1}{d_i}) d_{l}}\\
     &  = 2 e^{-\frac{1}{3} \frac{1}{(d_i -1)d_i}  d_{l}}\\
     &  \leq 2 e^{- \log^{2}n}\\
     &  \leq \frac{2}{n^{\log n}}\\
\end{align*}

where the second to last inequality is due to $d_l \geq d_{i+1} = 3 \log^{2} n \cdot d_i^2$. Therefore by a union bound on the edges of the matching in $P_l$ we get that 
\begin{align*}
    \P(\exists \ j = 1, \ldots k_l, \ \ e_j \subset S) &  \leq \frac{2k_l}{n^{\log n}} \\
    & \leq  \frac{1}{n^{\log n - 1}}\\
\end{align*}

Finally, another union bound on all the partition $l \geq i+1$ yields that the probability that there exists a partition $P_l$ such that an edge of $P_l$ is included in $S$ is less than $O(\log \log n / n^{\log n - 1}) = 1 - e^{-\Omega(\log^2 n)}$.
\end{proof}


\begin{lemma}\label{lemma:induction_loglogn_hardness}
If vertices in $P_{i}, P_{i+1},  \ldots, P_R$ are indistinguishable to $\A$ at the beginning of round $i$ of queries, then vertices in $P_{i+1},  \ldots, P_R$ are indistinguishable at the end of round $i$ with probability $1 - e^{-\Omega(\log^2 n)}$.
\end{lemma}

\begin{proof}
 Consider a query $S$ by an algorithm at round $i$ such that all the queries from previous rounds are independent of the partition $P_{i},\ldots, P_R$.
 Lemma \ref{lemma:small_query} shows that if $S$ is small ($|S| \leq (1-1/d_i)\sum_{j \geq i} |P_j|)$ and is independent of partition of $\cup_{j = i}^R P_i$ into $P_i, \ldots, P_R$, then with probability $1 - e^{-\Omega(\log^2 n)}$, for every $j \geq i+1$, $S$ does not contain any edge contained in $P_j$. On the other hand, Lemma \ref{lemma:large_query} shows that if a query $S$ has large size ($|S| \geq (1-1/d_i)\sum_{j \geq i} |P_j|)$ and is independent of partition of $\cup_{j = i}^R P_i$ into $P_i, \ldots, P_R$, then with probability $1 - e^{-\Omega(\log^2 n)}$, $S$ contains at least one edge from the matching on $P_i$ which implies $Q_{M_P}(S) = 1$. In both cases, $Q_{M_P}(S)$ is independent of the partition of $\cup_{j = i+1}^R P_i$ into $P_{i+1}, \ldots, P_R$ with probability $1 - e^{-\Omega(\log^2 n)}$. By a union bound, this holds for $poly(n)$ queries at round $i$.
\end{proof}

\thmhardnesslog*
\begin{proof}
Consider a uniformly random partition $P = (P_0, \ldots, P_i, \ldots P_{R})$, a matching $M_P$ and an algorithm $\mathcal{A}$ which queries $M_P$ in $\log \log n -3$ rounds. By Lemma \ref{lemma:induction_loglogn_hardness}, after $i$ rounds of queries, with probability $1 - e^{-\Omega(\log^2 n)}$ over both the randomization of $P$ and of the algorithm, all the queries $Q_{M_P}(S)$ made by $\mathcal{A}$ are independent of the partition of $P_i \cup \ldots \cup P_{R}$. Therefore, and since $R \geq \log \log n - 1$ by Claim \ref{claim:size_of_partition}, we get that after $\log \log n -3$ round of queries, with probability $1 - e^{-\Omega(\log^2 n)}$ over both the randomization of $P$ and the algorithm $\A$, all the queries $Q_{M_P}(S)$ made by $\A$ are independent
of the partition of $P_{R-1} \cup P_R$. We now distinguish two cases:
\begin{itemize}
    \item $\A$ does not a return any edge that is included in $P_{R-1} \cup P_R$.
    \item $\A$ returns a set of edges that is included in $P_{R-1} \cup P_R$. In this case, we know that with probability $1 - e^{-\Omega(\log^2 n)}$ all the queries $Q_{M_P}(S)$ made by $\A$ are independent of the partition of $P_{R-1} \cup P_R$ into $P_{R-1}$ and $P_R$.
    The edges that are returned by $\A$ and included in $P_{R-1} \cup P_R$ are therefore also independent from the partition of $P_{R-1} \cup P_R$ into $P_{R-1}$ and $P_R$ with probability $1 - e^{-\Omega(\log^2 n)}$. To fully learn $M_P$, $\A$ needs to make the distinction between points in $P_{R-1}$ and points in $P_{R}$, but there are $\binom{|P_R| + |P_{R-1}|}{|P_{R-1}|}$ ways of partitioning $P_{R-1} \cup P_R$ into $P_{R-1}$ and $P_R$. Therefore the probability that $\A$ correctly learns $M_P$ is less than
    \[ (1 - e^{-\Omega(\log^2 n)})\frac{1}{\binom{|P_R| + |P_{R-1}|}{|P_{R-1}|}}\]
    In the rest of the proof, we show that $1/\binom{|P_R| + |P_{R-1}|}{|P_{R-1}|} = e^{-\Omega(\log^2 n)}$. This implies that the probability that $\A$ learns $M_P$ correctly is less than $(1 - e^{-\Omega(\log^2 n)})e^{-\Omega(\log^2 n)}= e^{-\Omega(\log^2 n)}$. 
    
    We know that $n \leq \sum_{i = 0}^R |P_i| \leq (R+1)|P_R|$. Therefore $|P_R| \geq n/(R+1) = \Omega( n/\log\log n)$. By Claim \ref{claim:binomial_bounds}, we get that 
    \begin{align*}
        \frac{1}{\binom{|P_R| + |P_{R-1}|}{|P_{R-1}|}} & \leq \frac{4(|P_{R-1}|)!}{(|P_R| + |P_{R-1}|)^{|P_{R-1}|}}\\
        & \leq \left(\frac{|P_{R-1}|}{|P_R| + |P_{R-1}|}\right)^{|P_{R-1}|}\\
        & \leq \left(\frac{1}{3 \log^2 n |P_{R-1}|}\right)^{|P_{R-1}|}\\
        & \leq \left(\frac{1}{3 \log^2 n |P_{R-1}|}\right)^{3 \log^2 n}\\
        & \leq \left(\frac{1}{3 \log^2 n |P_{R-1}|}\right)^{3 \log^2 n}\\
        & = e^{-\Omega(\log^2 n)}
    \end{align*}
\end{itemize}

Therefore, with probability $1 - e^{-\Omega(\log^2 n)}$, the matching returned by $\mathcal{A}$ is not equal to $M_P$.
\end{proof}

\section{Missing Analysis for Low Degree Near Uniform Hypergraphs (Section~\ref{sec:bounded_degree})}
\label{appendix:hypergraphs}

\subsection{Proof of Lemma \ref{lemma:failure-prob-bn}}\label{appx:lemma_failure_prob}

\begin{lemma}\label{lemma:success_probability_bounded}
Assume that the hypergraph $H(V,E)$ has a maximum degree of $\Delta$ and edge size between $d/\rho$ and $d$. Let $S$ be a vertex-sample with probability $p$, and assume that $S$ contains an edge $e$. We have
\[ \P(\not\exists \ e' \mbox{ s.t } e' \subseteq S, e' \in H \setminus \{e\} \ | \ e \subseteq S) \geq f^{\bullet}(\Delta, p,d,\rho) \]
\end{lemma}

\begin{proof}
The intuition is that the term $f^{\bullet}(\Delta,p,d,\rho)$ treats the events that the edges are not contained in $S$ as independent events. \\

Consider the following intermediate optimization problem

\begin{alignat*}{3}
 f^\bullet_k(\Delta, p,d,\rho):=\quad\quad & \min\limits_{a_{ij}} & \prod\limits_{j=d/\rho}^{d}\prod\limits_{i=0}^{j-1} (1-p^{j-i})^{a_{ij}} \\
 &  \text{s.t.} \quad & \sum_{j = d/\rho}^{d}\sum_{i = 0}^{j-1} i\cdot  a_{ij} \leq \min\{(\Delta -1)d,kd\}\\
& & \sum_{j = d/\rho}^{d}\sum_{i = 0}^{j-1} j\cdot  a_{ij} \leq k\\
& & a_{ij} \geq 0.
\end{alignat*}


We want to show that for every $ 1 \leq k \leq \Delta n$, and for any $k$ edges $e_1,\ldots, e_k$, all different than $e$, we have
\[ \P(e_1, \ldots, e_k \not\subseteq S \ | \ e \subseteq S)\geq f_k^{\bullet}(\Delta,p,d,\rho).\]

For $k = 1$, the result clearly holds. Suppose by induction for any $k$ edges $e_1,\ldots, e_k$, all different than $e$, we have
\[ \P(e_1, \ldots, e_k \not\subseteq S \ | \ e \subseteq S)\geq f_k^{\bullet}(\Delta,p,d,\rho).\]
If we consider $k+1$ edges, then 
\begin{align*}
    \P(e_1, \ldots, e_k, e_{k+1} \not\subseteq S \ | \ e \subseteq S) & = P(e_1, \ldots, e_k \not\subseteq S \ | \ e \subseteq S) \P(e_{k+1} \not\subseteq S \ | \ e \subseteq S, \  e_1, \ldots, e_k, e_{k+1} \not\subseteq S) 
\end{align*}

By the induction hypothesis we know that $P(e_1, \ldots, e_k \not\subseteq S \ | \ e \subseteq S) \geq f_k^{\bullet}(\Delta,p,d,\rho)$. Therefore, if we can show that 
\begin{equation}\label{eq:inequality_toshow}
 \P(e_{k+1} \not\subseteq S \ | \ e \subseteq S, \  e_1, \ldots, e_k \not\subseteq S) \geq \P(e_{k+1} \not\subseteq S\ | \ e \subseteq S)
 \end{equation}
then we will have
\[ \P(e_1, \ldots, e_k, e_{k+1} \not\subseteq S \ | \ e \subseteq S) \geq f_{k}^{\bullet}(\Delta,p,d,\rho) \P(e_{k+1} \not\subseteq S\ | \ e \subseteq S) \geq  f_{k+1}^{\bullet}(\Delta,p,d,\rho)\]

To see that \eqref{eq:inequality_toshow} holds, we omit conditioning on $e \subset S$ to ease notation. By Bayes rule, 
\begin{align*}
   \P(e_{k+1} \not\subseteq S \ | \ e \subseteq S, \  e_1, \ldots, e_k \not\subseteq S) & = \P( e_{k+1} \not\subseteq  S \ |  \ \exists \ \ell \in [1,k]  \ e_{\ell} \cap e_{k+1} \setminus S \neq \emptyset,  \ e_1, \ldots, e_k \not\subseteq S)\\
   & \ \ \ \ \times \P(  \exists \ \ell \in [1,k] \ e_{\ell} \cap e_{k+1} \setminus S \neq \emptyset \ |  \ e_1, \ldots, e_k \not\subseteq S)\\ & \ \ \ \ + \P( e_{k+1} \not\subseteq  S \ |  \ \forall \ \ell \in [1,k] \ e_{\ell} \cap e_{k+1} \setminus S= \emptyset,\ e_1,\ldots, e_k \not\subseteq S)\\
   & \ \ \ \ \times \P( \forall \ \ell \in [1,k] \ e_{\ell} \cap e_{k+1} \setminus S= \emptyset \ |  \ e_1,\ldots, e_k \not\subseteq S \not\subseteq S)\\
\end{align*}

When $e_{\ell} \cap e_{k+1} \setminus S \neq \emptyset$ for some value $\ell$, then we know with probability one that $e_{k+1} \not\subseteq S$. Therefore
\[ \P( e_{k+1} \not\subseteq  S \ |  \ \exists \ \ell \in [1,k]  \ e_{\ell} \cap e_{k+1} \setminus S \neq \emptyset,  \ e_1, \ldots, e_k \not\subseteq S) = 1 \geq P(e_{k+1} \not\subseteq  S).\]

Furthermore,  \[\P( e_{k+1} \not\subseteq  S \ |  \ \forall \ \ell \in [1,k] \ e_{\ell} \cap e_{k+1} \setminus S= \emptyset,\ e_1,\ldots, e_k \not\subseteq S)= \P( e_{k+1} \not\subseteq  S)= 1 - p^d,\] 
therefore
\begin{align*}
   \P(e_{k+1} \not\subseteq S \ | \ e \subseteq S, \  e_1, \ldots, e_k \not\subseteq S) & \geq \P( e_{k+1} \not\subseteq  S \ |  \ \exists \ \ell \in [1,k]  \ e_{\ell} \cap e_{k+1} \setminus S \neq \emptyset,  \ e_1, \ldots, e_k \not\subseteq S)\\
   & \ \ \ \ \times  P(e_{k+1} \not\subseteq  S)\\ & \ \ \ \ + \P( e_{k+1} \not\subseteq  S \ |  \ \forall \ \ell \in [1,k] \ e_{\ell} \cap e_{k+1} \setminus S= \emptyset,\ e_1,\ldots, e_k \not\subseteq S)\\
   & \ \ \ \ \times  P(e_{k+1} \not\subseteq  S)\\
   & =  P(e_{k+1} \not\subseteq  S)
\end{align*}

\end{proof}

\begin{lemma}\label{lemma:aij-ai}
For $p \in (0,1)$, any optimal solution to $LP^\bullet(\Delta, p, d, \rho)$ and $f^\bullet(\Delta, p, d, \rho)$ is such that $a_{ij} = 0$ for all $j \neq d/\rho$.
\end{lemma}
\begin{proof}
For ease of notation, we use $d^-$ instead of $d/\rho$ in the proof. Suppose to the contrary that there is an optimal solution $a$ such that $a_{ij} > 0$ for some $j > d^-$. Then consider the alternative solution $a'$ such that $a'_{id^-} = a_{id^-} + \epsilon$, $a'_{ij} = a_{ij} - \epsilon$, and $a'_{sk} = a'_{sk}$ for all $s \neq i$ or $k \notin \{d^-, j\}$. For $\epsilon > 0$ small enough, clearly the constraints still hold. 

The net increase in objective value of $LP^\bullet(\Delta, p, d, \rho)$ is  $\epsilon\cdot (p^{d^- - 1} - p^{j - 1}) > 0$ because $d^- < j$ and $p \in (0,1)$.
The new objective value of $f^\bullet(\Delta, p, d, \rho)$ is $\left(\frac{1-p^{d^--i}}{1-p^{j-i}}\right)^\epsilon < 1$ times the old objective value (thus smaller than the old objective value) again because $d^- < j$ and $p \in (0,1)$. Therefore, $a$ is not an optimal solution to either $LP^\bullet(\Delta, p, d, \rho)$ or $f^\bullet(\Delta, p, d, \rho)$.
\end{proof}

From Lemma~\ref{lemma:aij-ai}, we can obtain the following two equivalent representations of $LP^\bullet(\Delta, p, d, \rho)$ and $f^\bullet(\Delta, p, d, \rho)$: let $a_i$ be the number of edges of size $d/\rho$ that intersect the focal edge at $i$ vertices, we have

\begin{alignat*}{3}
 LP^\bullet(\Delta, p, d, \rho):=\quad\quad & \max\limits_{a_{i}} & \sum_{i = 0}^{d/\rho-1} a_{i} \cdot p^{d/\rho-i} \\
 &  \text{s.t.} \quad & \sum_{i = 0}^{d/\rho-1} i\cdot  a_{i} \leq (\Delta -1)d\\
& & \sum_{i = 0}^{d/\rho-1} a_{i} \leq \frac{\Delta n}{d/\rho}\\
& & a_{i} \geq 0.
\end{alignat*}

\begin{alignat*}{3}
 f^\bullet(\Delta, p, d, d/\rho):=\quad\quad & \min\limits_{a_{i}} & \prod\limits_{i=0}^{d/\rho-1} (1-p^{d/\rho-i})^{a_{i}} \\
 &  \text{s.t.} \quad & \sum_{i = 0}^{d/\rho-1} i\cdot  a_{i} \leq (\Delta -1)d\\
& & \sum_{i = 0}^{d/\rho-1} a_{i} \leq \frac{\Delta n}{d/\rho}\\
& & a_{i} \geq 0.
\end{alignat*}

\begin{claim}\label{claim:bdf}
The function 
 $$f(x) = \left(\frac{1-p^d}{1-p^{d-x}}\right)^{\frac{1}{x}}$$ is increasing in $x$ for $x \in [1, d-1]$.
\end{claim}
\begin{proof}
Taking the derivative of $f$, we get 
\begin{equation} \label{eq:dbf'}
    f'(x) = \frac{\left(\frac{p^x(p^d-1)}{p^d-p^x}\right)^{\frac{1}{x}}}{(p^d-p^x)x^2}\cdot \left(p^d x \log p - (p^d - p^x)\log\left(\frac{p^x(p^d-1)}{p^d-p^x}\right)\right).
\end{equation}
The first term on the RHS of Eq.\eqref{eq:dbf'} is non-positive: $p^d < p^x$ and $p^d - 1 < 0$ because $p\in (0,1)$ and $x < d$. We now show that the second term is non-positive as well. Denote the second term by $g(x)$, i.e., 
$$g(x) = p^d x \log p - (p^d - p^x)\log\left(\frac{p^x(p^d-1)}{p^d-p^x}\right).$$ Our plan is to first show that $g(x)$ is decreasing in $x$, and therefore it is maximized at $x = 1$. We then show that $g(1) \leq 0$, and thus conclude that $g(x) \leq 0$ for all $x \in [1, d-1]$.

Taking derivative of $g$, we get
$$g'(x) = p^x\log p \log \left(\frac{p^x(p^d-1)}{p^d-p^x}\right) = p^x\log p \log \left(\frac{1-p^d}{1-p^{d-x}}\right).$$
Because $1 \leq d-x \leq d-1$ and $p \in (0,1)$, we have $0 \geq 1-p^d \geq 1-p^{d-x}$. Thus, $(1-p^d)/(1-p^{d-x})\geq 1$. Furthermore, $\log p < 0$ because $p \in (0,1)$. We can thus conclude that $g'(x) \leq 0$ for all $x \in [1, d-1]$, or equivalently, $g(x)$ is decreasing on $[1, d-1]$.

We now show $g(1) \leq 0$.
\begin{align*}
g(1) &= p^d \log p - (p^d - p) \log \left(\frac{p(p^d-1)}{p^d-p} \right)\\
&= p^d\left(\log p - \log \left(\frac{p(p^d-1)}{p^d-p} \right)\right) + p \log \left(\frac{p(p^d-1)}{p^d-p} \right)\\
& = p^d \log p - (p^d - p) \log \left(\frac{p^d-1}{p^{d-1}-1} \right)\\
& = p \left(p^{d-1}\log p - (p^{d-1} - 1) \log \left(\frac{1-p^d}{1-p^{d-1}} \right))\right)\\
& = p \left(p^{d-1}\log p + (1-p^{d-1}) \log \left(1+\frac{p^{d-1}-p^d}{1-p^{d-1}} \right)\right)\\
& \leq p \left(p^{d-1}\log p + (1-p^{d-1}) \frac{p^{d-1}-p^d}{1-p^{d-1}} \right)\\
& = p \left(p^{d-1}\log p + p^{d-1}-p^d \right)\\
& = p^d \left(\log p + 1-p \right) \leq 0,
\end{align*}
where the last inequality follows from the fact that $1+t \leq e^t$ for all $t$: we can set $t = \log p$. 
\end{proof}

\begin{lemma} \label{lemma:closed-form-nb}
For $p \in (0, 1)$, an optimal solution to the math program $f^\bullet(\Delta, p, d, \rho)$ is 
\[a^*_{d/\rho-1}= \min \left \{(\Delta - 1)\frac{d}{d/\rho-1},  \frac{\Delta n}{d/\rho} \right \}, \; a^*_0 = \frac{\Delta n}{d/\rho} - a^*_{d/\rho-1}, \; a^*_i = 0 \; \forall i \in \{ 1,\ldots, d/\rho-2 \}. 
\]
\end{lemma}
\begin{proof}
For ease of notation, we use $d^-$ instead of $d/\rho$ in the proof. Consider another solution $a$ such that $a_{d^--1} < \min\{(\Delta - 1)d/(d^--1), \Delta n/d^-\}$. If for all $i \in \{0, \ldots, d^--2\}$, $a_i = 0$, then $a$ is not optimal: neither constraint is tight; thus one could always increase some $a_i, i < d^--1$ to decrease the objective value. Therefore, without loss of generality, we can assume that there exists an index $i < d^--1$, $a_i > 0$. Let $i$ be the biggest such index. For $\epsilon > 0$, consider the following alternative solution:
$$a'_{d^--1} = a_{d^--1} + \epsilon, \; 
a'_i = \begin{cases} 
      a_i - \frac{{d^--1}}{i} \epsilon & i > 0 \\
      a_i - \epsilon & i = 0
   \end{cases}, $$
$$a'_0 = \begin{cases} 
      a_0 + \left(\frac{{d^--1}}{i} -1 \right)\epsilon & i > 0 \\
      a_0 - \epsilon & i = 0
   \end{cases}, \;
a'_k = a_k \;\forall k \notin \{0, i, {d^--1}\}
   $$

It can be easily checked that $a'$ is still feasible for small enough $\epsilon$. Let $obj(a)$ (resp. $obj(a')$) be the objective function value with respect to solution $a$ (resp. $a'$). When $i > 0$, we have
\begin{align*}
\frac{obj(a')}{obj(a)} &= \left( \frac{(1-p)\cdot (1-p^{d^-})^{\frac{d^--1}{i} - 1}}{(1-p^{d^--i})^{\frac{d^--1}{i}}}\right )^\epsilon\\
& = \left(\frac{1-p}{1-p^{d^-}} \left(\frac{ 1-p^{d^-}}{1-p^{d^--i}}\right)^{\frac{d^--1}{i}}\right )^\epsilon\\
& = \left(\left(\frac{1-p}{1-p^{d^-}}\right)^{\frac{1}{d^--1}} \left(\frac{ 1-p^{d^-}}{1-p^{d^--i}}\right)^{\frac{1}{i}}\right )^{(d^--1)\epsilon}\\
& = \left(\frac{\left(\frac{ 1-p^{d^-}}{1-p^{d^--i}}\right)^{\frac{1}{i}}}{\left(\frac{1-p^{d^-}}{1-p}\right)^{\frac{1}{d^--1}}}\right )^{(d^--1)\epsilon}\leq 1,
\end{align*}
where the last inequality follows from Claim~\ref{claim:bdf}. 
We have thus shown that $obj(a') \leq obj(a)$ if $i > 0$. 

We now show that $obj(a') \leq obj(a)$ if $i = 0$.
When $i = 0$, we have
\begin{align*}
\frac{obj(a')}{obj(a)} &= \left( \frac{1-p}{1-p^{d^-}}\right )^\epsilon \leq 1,
\end{align*}
because $0 < p^{d^-} < p$.
\end{proof}

\begin{lemma} \label{lemma:exp-nb}
For $\Delta \geq 2$, we have
\[ f^\bullet(\Delta,p,d,\rho) \geq f^\bullet(2,p,d/\rho,1)^{(\Delta - 1)\rho}\]
\end{lemma}
\begin{proof}
For ease of notation, we use $d^-$ instead of $d/\rho$ in the proof. From Lemma~\ref{lemma:closed-form-nb}, we have that 
\begin{align} 
f^\bullet(\Delta, p, d, \rho) &= (1-p)^{ \min \left \{(\Delta - 1)\frac{d}{d^--1},  \frac{\Delta n}{d^-} \right \}} (1-p^{d^-})^{\frac{\Delta n}{d^-} - \min \left \{(\Delta - 1)\frac{d}{d^--1},  \frac{\Delta n}{d^-} \right \}} \notag\\
&\geq (1-p)^{(\Delta - 1)\frac{d}{d^--1}} (1-p^{d^-})^{\frac{\Delta n}{d^-} - (\Delta - 1)\frac{d}{d^--1}}. \label{eq:bd-1-n}
\end{align}
Specifically when $\Delta = 2$, we have 
\begin{align*}
f^\bullet(2, p, d^-, 1) &= (1-p)^{ \min \left \{\frac{d^-}{d^--1},  \frac{2 n}{d^-} \right \}} (1-p^{d^-})^{\frac{2 n}{d^-} - \min \left \{\frac{d^-}{d^--1},  \frac{2 n}{d^-} \right \}}\\
&= (1-p)^{\frac{d^-}{d^--1}} (1-p^{d^-})^{\frac{2 n}{d^-} - \frac{d^-}{d^--1}},\end{align*}
where the second equality follows because for 
$2 \leq d^- \leq n$, we have $\frac{d^-}{d^--1} \leq \frac{2n}{d^-}$.

We therefore have
\begin{equation}\label{eq:bd-2-n}
f^\bullet(2, p, d^-, 1)^{(\Delta-1)\frac{d}{d^-}} = (1-p)^{\frac{d^-}{d^--1}(\Delta-1)\frac{d}{d^-}} (1-p^{d^-})^{\left(\frac{2 n}{d^-} - \frac{d^-}{d^--1}\right)(\Delta-1)\frac{d}{d^-}}.
\end{equation}

From Eqs.\eqref{eq:bd-1-n} and \eqref{eq:bd-2-n}, we have that to show the desired inequality, we can equivalently show 
\begin{enumerate}
    \item $$(1-p)^{(\Delta - 1)\frac{d}{d^--1}} \geq (1-p)^{\frac{d^-}{d^--1}(\Delta -1) \frac{d}{d^-}}.$$
    \item $$(1-p^{d^-})^{\frac{\Delta n}{d^-} - (\Delta - 1)\frac{d}{d^--1}} \geq (1-p^{d^-})^{(\frac{2 n}{d^-} - \frac{d^-}{d^--1})(\Delta -1) \frac{d}{d^-}}.$$
\end{enumerate}
We first show 
$$(1-p)^{(\Delta - 1)\frac{d}{d^--1}} \geq (1-p)^{\frac{d^-}{d^--1}(\Delta -1) \frac{d}{d^-}}.$$
Since $1-p \in (0,1)$, we equivalently want to show 
$$(\Delta - 1)\frac{d}{d^--1} \leq \frac{d^-}{d^--1}(\Delta -1) \frac{d}{d^-},$$
which, after canceling common terms and rearranging clearly holds.

We then show
$$(1-p^{d^-})^{\frac{\Delta n}{d^-} - (\Delta - 1)\frac{d}{d^--1}} \geq (1-p^{d^-})^{(\frac{2 n}{d^-} - \frac{d^-}{d^--1})(\Delta -1) \frac{d}{d^-}}.$$
Again, because $1-p^{d^-} \in (0,1)$, we equivalently need to show 
$$\frac{\Delta n}{d^-} - (\Delta - 1)\frac{d}{d^--1} \leq (\frac{2 n}{d^-} - \frac{d^-}{d^--1})(\Delta -1) \frac{d}{d^-}.$$

Equivalently
$$\frac{\Delta n}{d^-} \leq \left((\frac{2 n}{d^-} - \frac{d^-}{d^--1})\frac{d}{d^-}+\frac{d}{d^--1}\right)(\Delta -1) = \frac{2n}{d^-}\frac{d}{d^-}(\Delta - 1).$$

Dividing both sides by $\frac{n}{d^-}(\Delta -1)$, we get equivalently 
$$\frac{\Delta}{\Delta-1} \leq 2\frac{d}{d^-},$$
which holds for $\Delta \geq 2$ because $\Delta/(\Delta - 1) \leq 2$ and $d/d^- \geq 1$.

\end{proof}

\begin{lemma} \label{lemma:f-lp-2}
Let $p> 0$, then we have
\[ f^\bullet(2,p,d,\rho) \geq 1 - LP^\bullet(2,p,d,\rho).\]
\end{lemma}

\begin{proof}

For ease of notation, we use $d^-$ instead of $d/\rho$ in the proof. Let $a = (a_{ij})_{j \in [d^-,d], \ i \in [0,j-1]}$ be a feasible solution to $LP^\bullet(2,p,d,\rho)$. It is sufficient to show that 
\[
  \prod\limits_{j=d^-}^{d}\prod\limits_{i=0}^{j-1} (1-p^{j-i})^{a_{ij}} \geq 1 - \sum_{j = d^-}^{d}\sum_{i = 0}^{j-1} a_{ij} \cdot p^{j-i}\]

We show, more generally that for $x_0, \ldots, x_k \in [0,1]^k$ and $a_0,\ldots, a_k \geq 0$
\[\prod\limits_{i=0}^{k} (1-x_i)^{a_i} \geq 1 - \sum_{i = 0}^{k} a_i \cdot x_i,\]

If $1-\sum_{i = 0}^{k} a_i \cdot x_i \leq 0$,  the inequality is trivially true. So we assume that $1-\sum_{i = 0}^{k} a_i \cdot x_i \geq 0$, which implies $1-a_i\cdot x_i \geq 0$ for $i \in [k]$. We first show that for a positive integer $a$ and $x \in [0,1]$ we have
\[ (1-x)^a \geq 1-ax\]
To see this, we consider the function $g(x) = (1-x)^a - (1-ax)$ over $[0,1]$. The derivative of $g$ is $g'(x) = a-a(1-x)^{a-1} \geq 0$ for $x \in [0,1]$. $g$ is therefore increasing and since $g(0) = 0$, we get that $g(x) \geq 0$ for $x \in [0,1]$. Therefore, for $i \in [k]$
\[ (1-x_i)^{a_i} \geq 1-a_i \cdot x_i \geq 0\]
By taking the product we get that 
\begin{equation}\label{eq:firstproduct}
\prod\limits_{i=0}^{k} (1-x_i)^{a_i} \geq \prod\limits_{i=0}^{k} (1-a_i \cdot x_i)
\end{equation}

All is left is to show that 
\begin{equation}\label{eq:induction}
    \prod\limits_{i=0}^{k} (1-a_i \cdot x_i) \geq 1 - \sum_{i = 0}^{k} a_i \cdot x_i
    \end{equation} 
We show \eqref{eq:induction} by induction. It is true for $i = 0$. Suppose \eqref{eq:induction} holds for some $\ell < k$, then 
\[\prod\limits_{i=0}^{\ell} (1-x_i)^{a_i} \geq 1 - \sum_{i = 0}^{\ell} a_i \cdot x_i, \]
and by multiplying by $(1-a_{\ell +1} \cdot x_{\ell + 1}) $, we get
\begin{align*}
\prod\limits_{i=0}^{\ell +1} (1-a_i \cdot x_i) & \geq (1-a_{\ell +1} \cdot x_{\ell + 1}) (1 - \sum_{i = 0}^{\ell} a_i \cdot x_i)\\
& =  1- \sum_{i = 0}^{\ell + 1} a_i \cdot x_i + (a_{\ell +1} \cdot x_{\ell + 1})(\sum_{i = 0}^{\ell + 1} a_i \cdot x_i)\\
& \geq  1 - \sum_{i = 0}^{\ell + 1} a_i \cdot x_i
\end{align*}
This concludes the proof of \eqref{eq:induction}. From \eqref{eq:firstproduct} and \eqref{eq:induction}, we get 
\[\prod\limits_{i=0}^{k} (1-x_i)^{a_i} \geq 1 - \sum_{i = 0}^{k} a_i \cdot x_i.\]
Finally, by setting $k = d^- -1$, and $x_i = p^{d^- - i}$, we get that 
\[\prod\limits_{i=0}^{d^--1} (1-p^{d^--i})^{a_i} \geq 1 - \sum_{i = 0}^{d^--1} a_i \cdot p^{d^--i},\]
which proves the lemma.
\end{proof}

\begin{claim}\label{claim:sufficient_bound_lp}
For $n\geq 100$ and $2 \leq d \leq n$, we have 
\[ \frac{d}{d-1}n^{-1/d}+\frac{1}{d} \leq \frac{n}{n-1}n^{-1/n}+\frac{1}{n} \leq 1 - \frac{\log n}{2n}\]
\end{claim}

\begin{proof}
We prove the claim in three steps.
\begin{enumerate}
    \item $\frac{d}{d-1}n^{-1/d}+\frac{1}{d} \leq \frac{n}{n-1}n^{-1/n} + \frac{1}{n}$ for $n \geq 42$. \\
    
    To see this, we analyze the derivative of 
$$q(w) = \frac{w}{w-1}\cdot n^{-\frac{1}{w}} + \frac{1}{w}, \;\; w \in [2,\infty).$$
We get
\begin{align*}
    q'(w) & = \frac{n^{-1/w}\cdot \left(-n^{1/w} + 2w n^{1/w} - w^2 (1+ n^{1/w}) + (w-1) w \log n\right)}{w^2(w-1)^2}\\
    & = \frac{n^{-1/w}\cdot \left(-n^{1/w}\cdot (w-1)^2 - w^2 + w(w-1)\log n\right)}{w^2(w-1)^2}\\
    & \geq \frac{n^{-1/w}\cdot \left(-e\cdot (w-1)^2 - w^2 + w(w-1)\log n\right)}{w^2(w-1)^2}.
    \end{align*}
    When $\log n \geq e + 1$, we have $w(w-1)\log n \geq w(w-1)(e+1) \geq w^2 + e(w-1)^2$, and therefore $q'(w) \geq 0$ when $n \geq 42 \geq e^{e+1}$.
    \item $n^{-1/n} \leq 1 - \frac{3}{4}\frac{\log n}{n}$ for $n \geq 10$\\
    
    To see this, consider the function $g(x) = 1 - \frac{3}{4}\frac{\log x}{x} - x^{-1/x}$ over the interval $[e,\infty)$. We have
    \[ g'(x) = \frac{1 - \log x}{x^2} (x^{-1/x} - \frac{3}{4}).\]
    Since $x^{-1/x}$ is increasing and $10^{-1/10} > 3/4$, we have that $g'(x) < 0$ for $x \geq 10$. Therefore, $g$ is decreasing over $[10,\infty]$. Since $\lim\limits_{x \rightarrow \infty} g(x) = 0$, this implies that $g(n) \geq 0$ for $n \geq 10$.
  
    \item $\frac{n}{n-1}( 1 - \frac{3}{4}\frac{\log n}{n})+\frac{1}{n} \leq 1 - \frac{\log n}{2n}$ for $n \geq 100$.\\
    
    To see this, we use the following sequence of equivalences
    \begin{align*}
        \frac{n}{n-1}( 1 - \frac{3}{4}\frac{\log n}{n})+\frac{1}{n} \leq 1 - \frac{\log n}{2n} & \Leftrightarrow (1+\frac{1}{n-1})( 1 - \frac{3}{4}\frac{\log n}{n})+\frac{1}{n} \leq 1 - \frac{\log n}{2n}\\
        & \Leftrightarrow \frac{1}{n-1} + \frac{1}{n}- \frac{3}{4}\frac{\log n}{n(n-1)}) \leq  \frac{\log n}{4n}\\
        &\Leftrightarrow 4n + 4(n-1) - 3\log n \leq  (n-1)\log n\\
    \end{align*}
    The last inequality is true for $n \geq 100$
\end{enumerate}
\end{proof}

\begin{lemma}\label{lemma:lp-general-closed-form}
The optimal solution of $LP^\bullet(2,p,d/\rho,1)$ is given by \begin{equation}\label{eq:bfs}
a_{d/\rho-1} = \min\left\{\frac{d/\rho}{d/\rho-1}, \frac{2n}{d/\rho}\right\}, \; a_0 = \frac{2n}{d/\rho} - a_{d/\rho-1}, \; a_i = 0 \;\; \forall i \in \{1, 2, \ldots, d/\rho-2\}.\end{equation}
\end{lemma}
\begin{proof}
For ease of notation, we use $d^-$ instead of $d/\rho$ in the proof. Suppose to the contrary that there exists an optimal solution $a$ such that $a_{d^--1} < \min\{d^-/(d^--1), 2n/d^-\}$. There must exist $i < d^--1$ such that $a_i > 0$: neither constraint is tight if for all $i < d^--1$, $a_i = 0$; thus one could always increase some $a_i, i < d^--1$ to increase the objective value. Let $i$ be the biggest index below $d^--1$ such that $a_i > 0$. For $\epsilon > 0$, consider the following alternative solution:
$$a'_{d^--1} = a_{d^--1} + \epsilon, \; 
a'_i = \begin{cases} 
      a_i - \frac{{d^--1}}{i} \epsilon & i > 0 \\
      a_i - \epsilon & i = 0
   \end{cases}, \;
a'_0 = \begin{cases} 
      a_0 + \left(\frac{{d^--1}}{i} -1 \right)\epsilon & i > 0 \\
      a_0 - \epsilon & i = 0
   \end{cases},$$
$$a'_k = a_k \;\;\forall k \notin \{0, i, {d^--1}\}
   $$

It can be easily checked that $a'$ is still feasible for small enough $\epsilon$. When $i > 0$, the net increase in objective value from $a$ to $a'$ is
\begin{align} \label{eq:netincrease}
\epsilon\cdot \left(p^{d^--{(d^--1)}} + \left(\frac{{d^--1}}{i}-1\right)\cdot p^d - \frac{{d^--1}}{i}\cdot p^{d-i}\right).
\end{align}

Now consider the following function $$f(x) = \frac{1}{x} (p^{d^--x}-p^{d^-}), \;\; x \in [1, {d^--1}].$$
We claim that $f(x)$ is strictly increasing when $p \in (0,1)$: taking the derivative of $f(x)$, we get
$$f'(x) = \frac{p^{d^--x}\cdot (p^x - (1+x\log p))}{x^2} = \frac{p^{d^--x}\cdot (e^{x\log p} - (1+x\log p))}{x^2} > 0,$$
where the inequality follows from the fact that $e^t > 1+t$ for all $t \neq 0$ and $p \in (0,1)$.

Therefore, we have that $f({d^--1}) > f(i)$, which translates to 
\begin{equation} \label{eq:kineq}\frac{1}{{d^--1}} (p^{d^--{(d^--1)}} -p^{d^-}) > \frac{1}{i} (p^{d^--i} -p^{d^-}).\end{equation}

By rearranging Eq.\eqref{eq:kineq}, we can conclude that the net increase given in \eqref{eq:netincrease} is strictly positive, thus contradicting the assumption that $a$ is an optimal solution.

Now consider the case where $i = 0$. The net increase in objective value from $a$ to $a'$ is 
\begin{align} \label{eq:netincrease-0}
\epsilon\cdot \left(p^{d^--{(d^--1)}} - p^{d^-}\right),
\end{align}
which is clearly strictly positive if $p \in (0,1)$, contradicting the assumption that $a$ is an optimal solution.
\end{proof}

\begin{lemma} \label{lemma:lp2ub-nb}
Assume $d \geq 2$. When $p = (2n)^{-\frac{\rho}{d}}$  and $n \geq 100$, we have $$LP^\bullet(2,p,d/\rho,1) \leq 1 - \frac{\log n}{2n}.$$
\end{lemma}
\begin{proof}
From Lemma~\ref{lemma:lp-general-closed-form}, we know that
$$LP^\bullet(2,p,d/\rho,1) \leq \frac{d/\rho}{d/\rho-1}\cdot p + \frac{2n}{d/\rho}\cdot p^{d/\rho}.$$

Therefore, when $p = (2n)^{-\frac{\rho}{d}}$, 
\begin{align}
LP^\bullet(2,p,d/\rho,1) & \leq \frac{d/\rho}{d/\rho-1}\cdot (2n)^{-\frac{1}{d/\rho}} + \frac{2n}{d/\rho}\cdot (2n)^{-1} \notag\\
& \leq \frac{d/\rho}{d/\rho-1} \cdot n^{-\frac{1}{d/\rho}} + \frac{1}{d/\rho}, \label{eq:dplus}
\end{align}
where the last inequality follows because $(2n)^{-\frac{1}{d/\rho}} \leq n^{-\frac{1}{d/\rho}}$.

From Claim \ref{claim:sufficient_bound_lp}, we have that for $d/\rho \geq 2$ and $n \geq 100$, \eqref{eq:dplus} is upper bounded by $1 - \frac{\log n}{2n}$.
\end{proof}

\subsection{Proof of Theorem \ref{thm:main}}\label{appx:proof_main_thm}

\thmmain*
\begin{proof}
To show that Algorithm \ref{alg:1round-boundeddegree-nu} returns the hypergraph with high probability, we will show the following:
\begin{enumerate}
    \item With high probability, for every edge, there exists at least one sample that contains that edge and no other edges.
    \item If a sample contains more than one edge, we will learn at most the intersection of the edges, and this intersection will be discarded in the last for-loop of the algorithm
\end{enumerate}

The second point above is easy to see: if a sample $S_i$ contains more than one edge, then $Q_H(S_i \setminus \{v\}) = 0$ if and only if $v$ is in the intersection of all edges contained in $S_i$. If we can successfully learn any edge in $S_i$ from another sample $S_j$ that contains only that edge, then the intersection learned through sample $S_i$ will be discarded as it is a subset of the edge learned using $S_j$. 

For the rest of the proof, we focus on showing the first point.
To reiterate, we want to show that every edge of $H$ appears in at least one sample by itself. Below, we show that it is the case with probability $1-o(1)$.

We first show that with high probability, each edge $e$ of size $d' \in [d/\rho, d]$ is contained in at least $(2n)^{\Delta - 1}(\log^2 n)/2$ sample $S_i$'s: use $X_e$ to denote the number of samples containing $e$, then we have
\begin{align*}
\E(X_e) &= (2n)^{\Delta \rho}\log^2 n \cdot p^{d'}\\
&\geq  (2n)^{\Delta \rho}\log^2 n \cdot p^{d}\\
&= (2n)^{\Delta \rho}\log^2 n \cdot (2n)^{-\rho}\\
&= (2n)^{(\Delta - 1)\rho}\log^2 n
\end{align*}

By Chernoff bound, we have
$$\P(X_e \leq (2n)^{(\Delta - 1)\rho}(\log^2 n)/2) \leq e^{-(2n)^{(\Delta - 1)\rho}(\log^2 n)/8}.$$
As there are at most $\Delta n \leq n^2$ edges of size $d' \in [d/\rho, d]$, by a union bound, we have
\begin{align*}
& P(\exists e \in E \text{ s.t. } |e|\in [d/\rho, d], X_e \leq (2n)^{(\Delta - 1)\rho}(\log^2 n)/2)\\ 
& \leq \Delta n e^{-(2n)^{(\Delta - 1)\rho}(\log^2 n)/8} = e^{\log \Delta + \log n -(2n)^{(\Delta - 1)\rho}(\log^2 n)/8}\\
& \leq e^{-\log n ((2n)^{(\Delta - 1)\rho}(\log n)/8-2)}\\
& \leq e^{-\log n((\log n)/8-2)} = n^{-(\log n)/8+2} = o(1).
\end{align*}
Subsequently, 
$$P(\forall e \in E \text{ s.t. } |e|\in [d/\rho, d], X_e \geq (2n)^{(\Delta - 1)\rho}(\log^2 n)/2) \geq 1- o(1).$$

From now on we condition on the event that all edges $e$ whose size is between $d/\rho$ and $d$ are included in at least $(2n)^{(\Delta - 1)\rho}(\log^2 n)/2$ samples. 

From Lemma~\ref{lemma:failure-prob-bn}, we have that for all $n \geq 100$, 
\begin{align*}
\P(\exists e' \in E \setminus \{e\}, \ e' \subseteq S_i \ | \ e \subseteq S_i) \leq 1-(\frac{\log n}{2n})^{(\Delta - 1)\rho}.
\end{align*}

As each edge $e$ with size between $d/\rho$ and $d$ is contained in at least $(2n)^{(\Delta - 1)\rho}(\log^2 n)/2$ samples, we have that 
\begin{align*}
\P(\forall S_i \text{ s.t. } e \in S_i, \exists e' \in E \setminus \{e\} \mbox{ s.t. } e' \subseteq S_i)
& \leq \left(1-\left(\frac{\log n}{2n}\right)^{(\Delta - 1)\rho}\right)^{(2n)^{(\Delta - 1)\rho}(\log^2 n)/2}\\
& \leq e^{-\left(\frac{\log n}{2n}\right)^{(\Delta - 1)\rho}(2n)^{(\Delta - 1)\rho}(\log^2 n)/2}\\
& = e^{-(\log^{(\Delta + 1)\rho}n)/2} = n^{-(\log^{\Delta\rho} n)/2}.
\end{align*}


By another union bound on all edges of size between $d/\rho$ and $d$ (at most $\Delta n$ of them), we have that 
\begin{align*}
\P(\exists e \text{ s.t. } |e| \in [d/\rho, d], \forall S_i \text{ s.t. } e \in S_i, \exists e' \in E \setminus \{e\}, \ e' \subseteq S_i) \leq n^{-(\log^{\Delta\rho}n)/2+2} = o(1). 
\end{align*}
We can thus conclude that with probability at least $1-o(1)$, for all $e \in E$ with size between $d/\rho$ and $d$, there exists at least one sample $S_i$ that contains $e$ but no other remaining edges.\\

We now prove the query complexity of \textsc{FindLowDegreeEdges}. It is clear that \textsc{FindLowDegreeEdges} is non adaptive because all the queries can be made in parallel. Regarding query complexity, \textsc{FindLowDegreeEdges} constructs  $(2n)^{\rho\Delta}\log^2 n $ samples, and for each one of these samples, makes at most $\O(n)$ queries. Therefore \textsc{FindLowDegreeEdges} makes at most 
\[ \O(n (2n)^{\rho\Delta}\log^2 n) = \O((2n)^{\rho\Delta+1}\log^2 n)\]
queries in total.
\end{proof}

\section{Sequential and Parallel Runtime}
The runtime of our proposed algorithms is not much worse than their query complexity. When running sequentially, \textsc{FindEdgeAdaptive} requires an additional $O(n)$ time to do every set partition, \textsc{FindDisjointEdges} requires an additional $O(n)$ time to construct a set of independently sampled vertices, and \textsc{FindLowDegreeEdges} requires an additive $(\Delta n)^2 d$ time to execute the last for loop that consolidates the candidate edge sets. We summarize both the sequential and parallel runtimes for each algorithm in Table \ref{tab:runtime} below which we will include in the next version of the paper. We assume that each query can be made in $O(1)$ time, and for the parallel runtime analysis, we assume access to polynomially many parallel processors. Some of our algorithm use a subroutine that can either be  \textsc{FindEdgeParallel} or \textsc{FindEdgeAdaptive}. We will use \textsc{PRL} and \textsc{ADA} to refer to these subroutines respectively.

\begin{table}[h]
\centering
\begin{tabular}{|p{55mm}|p{25mm}|p{30mm}|p{25mm}|} 
\hline
Algorithm & Query \newline Complexity & Sequential \newline Runtime & Parallel \newline Runtime \\ \hline
     $\textsc{PRL: FindEdgeParallel}(S)$     &    $O(|S|)$              &         $O(|S|)$           &          $O(1)$        \\ \hline
     $\textsc{ADA: FindEdgeAdaptive}(S, s)$     &  $O(s \log(|S|))$               &      $O(s|S|\log(|S|))$              &         $O(\log|S|)$         \\ \hline
     $\textsc{FindDisjointEdges}(\textsc{PRL})$     &     $O(n^\alpha \log^2 n$ \newline $+n^{2+\alpha} \log^2 n)$             &        $O(n^{\alpha+1} \log^2 n)$            &        $O(1)$          \\ \hline
     $\textsc{FindDisjointEdges}(\textsc{ADA})$     &    $O(n^\alpha \log^2 n$\newline $+ 2n^\alpha \log^3 n)$              &        $O(n^{\alpha+2} \log^3 n)$            &           $O(\log n)$       \\ \hline
     $\textsc{FindMatching}(\textsc{PRL})$     &      $O(n^3 \log^3 n)$            &      $O(n^3 \log^3 n)$              &       $O(\log^2 n)$           \\ \hline
     $\textsc{FindMatching}(\textsc{ADA})$     &     $O(n \log^5 n)$             &       *$O(n^{\alpha +2} \log^6 n)$             &        $O(\log^4 n)$          \\ \hline
     $\textsc{FindLowDegreeEdges}$     &       $O((2n)^{\rho \Delta + 1}$\newline $\log^2 n)$           &     $O((2n)^{\rho \Delta + 1} \log^2 n$ \newline $+(\Delta n)^2 d)$               &       $O(\log n)$           \\ \hline
\end{tabular}
\caption{Algorithms Runtime (with $\alpha = 1 / (1 - 1/(2\log n))$ in the cell with *)}
\label{tab:runtime}
\end{table}

\end{document}